\documentclass[11pt]{article}
\usepackage[utf8]{inputenc}

\usepackage{hdb_macros}
\usepackage{fullpage}
\usepackage{mathtools}
\usepackage{newtxtext}
\usepackage{microtype}
\usepackage[normalem]{ulem}

\usepackage[multiuser,inline,nomargin]{fixme}
\FXRegisterAuthor{y}{ey}{\color{teal}Yi}
\FXRegisterAuthor{c}{ec}{\color{blue}Chris}
\FXRegisterAuthor{h}{eh}{\color{purple}Huck}
\fxusetheme{color}

\newcommand{\mat}[1]{#1} %

\newcommand{\alphaeth}{\alpha^\ddagger}
\newcommand{\alphaseth}{\alpha^\dagger}
\newcommand{\kissnum}{c_\mathsf{kn}}
\newcommand{\alphakissnum}{\alpha_\mathsf{kn}}
\newcommand{\betafun}[3][p]{\beta_{#1,#2}(#3)}
\newcommand{\betainv}[3][p]{\beta_{#1,#2}^{-1}(#3)}

\newcommand{\sm}[1][p]{\lambda^{(#1)}}
\newcommand{\binGapCVP}{\problem{CVP}'}
\newcommand{\ratioGA}{400 \alpha}

\newcommand{\latdag}{\lat^{\dagger}}
\newcommand{\Bdag}{B^{\dagger}}
\newcommand{\tdag}{\vec{t}^{\dagger}}

\newcommand*\samethanks[1][\value{footnote}]{\footnotemark[#1]}

\renewcommand{\GapSVP}{\problem{SVP}}
\renewcommand{\GapCVP}{\problem{CVP}}

\title{Improved Hardness of BDD and SVP Under Gap-(S)ETH}
\author{Huck Bennett\thanks{University of Colorado Boulder. Email: \texttt{huck.bennett@colorado.edu}. The majority of this work was completed while the author was at the University of Michigan, and some was completed while the author was at Oregon State University.}\; \samethanks[3]%
    \and Chris Peikert\thanks{University of Michigan. Emails: \texttt{\{cpeikert,yit\}@umich.edu}.}\; \thanks{This material is based upon work supported by the National Science Foundation under Award CCF-2006857. The views expressed are those of the authors and do not necessarily reflect the official policy or position of the National Science Foundation.}
    \and Yi Tang\samethanks[2]\; \samethanks[3]
    }

\date{\today}

\begin{document}

\maketitle

\listoffixmes

\begin{abstract}
We show improved fine-grained hardness of two key lattice problems in the~$\ell_p$ norm: Bounded Distance Decoding to within an $\alpha$ factor of the minimum distance ($\BDD_{p, \alpha}$) and the (decisional) $\gamma$-approximate Shortest Vector Problem ($\GapSVP_{p,\gamma}$),
assuming variants of the Gap (Strong) Exponential Time Hypothesis (Gap-(S)ETH). Specifically, we show:
\begin{enumerate}
    \item \label{en:abs-improved-eth-kn-bdd}
    For all $p \in [1, \infty)$, there is no $2^{o(n)}$-time algorithm for $\BDD_{p, \alpha}$ for any constant $\alpha > \alphakissnum$, where $\alphakissnum = 2^{-\kissnum}$ and $\kissnum$ is the $\ell_2$ kissing-number constant, assuming $\kissnum > 0$ and that non-uniform Gap-ETH holds.
    
    \item \label{en:abs-improved-eth-Zn-bdd}
    For all $p \in [1, \infty)$, there is no $2^{o(n)}$-time algorithm for $\BDD_{p, \alpha}$ for any constant $\alpha > \alphaeth_p$, 
    where $\alphaeth_p$ is explicit and satisfies $\alphaeth_p = 1$ for $1 \leq p \leq 2$, $\alphaeth_p < 1$ for all $p > 2$, and $\alphaeth_p \to 1/2$ as $p \to \infty$, assuming that randomized Gap-ETH holds.
    
    \item \label{en:abs-improved-seth-bdd}
    For all $p \in [1, \infty) \setminus 2 \Z$ and all $C > 1$, there is no $2^{n/C}$-time algorithm for $\BDD_{p, \alpha}$ for any constant $\alpha > \alphaseth_{p, C}$, where $\alphaseth_{p, C}$ is explicit and satisfies $\alphaseth_{p, C} \to 1$ as $C \to \infty$ for any fixed $p \in [1, \infty)$, assuming $\kissnum > 0$ and that non-uniform Gap-SETH holds.
    
    \item \label{en:abs-improved-seth-svp}
    For all $p > p_0 \approx 2.1397$, $p \notin 2\Z$, and all $C > C_p$, there is no $2^{n/C}$-time algorithm for $\GapSVP_{p, \gamma}$ for some constant $\gamma > 1$, where $C_p > 1$ is explicit and satisfies $C_p \to 1$ as $p \to \infty$, assuming that randomized Gap-SETH holds.
\end{enumerate}

An earlier version of this work stated \cref{en:abs-improved-eth-kn-bdd,en:abs-improved-seth-bdd} without the assumption $\kissnum > 0$, which asserts that there exists a family of exponential kissing number lattices, since that claim had been made in prior published work. Unfortunately, that work was later shown to contain a bug, and so \cref{en:abs-improved-eth-kn-bdd,en:abs-improved-seth-bdd} are now conditional.

Our results for $\BDD_{p, \alpha}$ improve and extend work by Aggarwal and Stephens-Davidowitz (STOC, 2018) and Bennett and Peikert (CCC, 2020).
Specifically, the quantities~$\alphakissnum$ and~$\alphaeth_p$ (respectively, $\alphaseth_{p,C}$) (conditionally) improve upon the corresponding quantity~$\alpha_p^*$ (respectively,~$\alpha_{p,C}^*$) of Bennett and Peikert for small~$p$ (but arise from somewhat stronger assumptions).
Finally, \cref{en:abs-improved-seth-svp} answers a natural question left open by Aggarwal, Bennett, Golovnev, and Stephens-Davidowitz (SODA, 2021), which showed an analogous result for the Closest Vector Problem.

\end{abstract}

\newpage

\section{Introduction}

Lattices are geometric objects that look like regular orderings of points in real space. More formally, a lattice~$\lat$ is the set of all integer linear combinations of some linearly independent vectors $\vec{b}_1, \ldots, \vec{b}_n \in \R^m$. 
The matrix $B = (\vec{b}_1, \ldots, \vec{b}_n)$ whose columns are these vectors is called a \emph{basis} of~$\lat$, and we denote the lattice it generates by $\lat(B)$, i.e., $\lat = \lat(B) := \set{\sum_{i=1}^n a_i \vec{b}_i : a_1, \ldots, a_n \in \Z}$. The number of vectors~$n$ in a basis is an invariant of~$\lat$, and is called its \emph{rank}.

In recent years, lattices have played a central role in both cryptanalysis and the design of secure cryptosystems. One very attractive quality of many lattice-based cryptosystems (e.g.,~\cite{conf/stoc/Ajtai98,conf/stoc/AjtaiD97,journals/siamcomp/MicciancioR07,journals/jacm/Regev09,conf/stoc/GentryPV08}) is that they are secure assuming that certain key lattice problems are sufficiently hard to approximate in the \emph{worst case}.
Motivated by this and myriad other applications of lattices in computer science, many works have studied the $\NP$-hardness of both exact and approximate lattice problems (e.g.,~\cite{emde81:_anoth_np,journals/jcss/AroraBSS97,conf/stoc/Ajtai98,journals/siamcomp/Micciancio00,journals/tit/Micciancio01,journals/jcss/Khot06,journals/jacm/Khot05,conf/approx/LiuLM06,journals/toc/HavivR12,journals/toc/Micciancio12}).
More recently, motivated especially by the need to understand the \emph{concrete} security of lattice-based cryptosystems, a number of  works~\cite{conf/focs/BennettGS17,conf/stoc/AggarwalS18,conf/coco/BennettP20,journals/ipl/AggarwalC21,conf/soda/AggarwalBGS21} have studied the \emph{fine-grained} complexity of lattice problems. That is, for a meaningful real-world security bound, it is not enough to say merely that there is no \emph{polynomial-time} algorithm for a suitable lattice problem. Rather, a key goal is to show $2^{\Omega(n)}$-hardness, or even $2^{Cn}$-hardness for some explicit $C > 0$, of a problem under general-purpose complexity-theoretic assumptions, like variants of the (Strong) Exponential Time Hypothesis.

In this work, we extend the latter line of research by showing improved fine-grained complexity results for two key lattice problems,
the Bounded Distance Decoding Problem ($\BDD$) and the Shortest Vector Problem ($\GapSVP$). To define these problems, we first recall some notation.
Let $\lambda_1(\lat) := \min_{\vec{v} \in \lat \setminus \set{\vec{0}}} \norm{\vec{v}}$ denote the \emph{minimum distance} of $\lat$, i.e., the length of a shortest non-zero vector in $\lat$, and let $\dist(\vec{t}, \lat) := \min_{\vec{v} \in \lat} \norm{\vec{t} - \vec{v}}$ denote the distance between a target vector~$\vec{t}$ and~$\lat$. When using the~$\ell_p$ norm, we denote these quantities by $\sm_1(\lat)$ and $\dist_p(\vec{t}, \lat)$, respectively.

\paragraph{BDD and SVP.}
The Bounded Distance Decoding Problem in the $\ell_p$ norm for relative distance~$\alpha$, denoted $\BDD_{p, \alpha}$, is the search promise problem defined as follows: given a basis~$B$ of a lattice $\lat = \lat(B)$ and a target vector $\vec{t}$ satisfying $\dist_p(\vec{t}, \lat) \leq \alpha \cdot \sm_1(\lat)$ as input, the goal is to find a closest lattice vector $\vec{v} \in \lat$ to the target vector $\vec{t}$ such that $\norm{\vec{t} - \vec{v}}_p = \dist_p(\vec{t}, \lat)$. (We note that~$\vec{v}$ is guaranteed to be unique when $\alpha < 1/2$, but that $\BDD_{p, \alpha}$ is well-defined for any $\alpha = \alpha(n) > 0$.)
The $\gamma$-approximate Shortest Vector Problem in the $\ell_p$ norm, denoted $\GapSVP_{p, \gamma}$, is the decision promise problem defined as follows: given a basis~$B$ of a lattice~$\lat = \lat(B)$ and a distance threshold $r > 0$ as input, the goal is to decide whether $\sm_1(\lat) \leq r$ (i.e., the input is a YES instance) or $\sm_1(\lat) > \gamma r$ (i.e., the input is a NO instance), with the promise that one of the two cases holds.\footnote{In other literature, this decision problem is often called $\problem{GapSVP}_{p,\gamma}$, whereas $\problem{SVP}_{p,\gamma}$ usually denotes the corresponding \emph{search} problem (of finding a nonzero lattice vector $\vec{v} \in \lat$ for which $\norm{\vec{v}} \leq \gamma \cdot \sm_1(\lat)$, given an arbitrary lattice~$\lat$.) There is a trivial reduction from the decision problem to the search problem, so any hardness of the former implies identical hardness of the latter.}

Although it seems far out of reach using known techniques, proving that $\GapSVP_{p, \gamma}$ is hard for a sufficiently large polynomial approximation factor $\gamma = \gamma(n)$, or that $\BDD_{p, \alpha}$ is hard for sufficiently small inverse-polynomial relative distance $\alpha = \alpha(n)$, would imply the provable security of lattice-based cryptography.\footnote{We note that the relative distance~$\alpha$ in $\BDD_{p, \alpha}$ is not an approximation factor \emph{per se}, but it is analogous to one in a precise sense. Namely, for $p = 2$ there is a rank-preserving reduction from $\BDD_{2, \alpha}$ to $\GapSVP_{2, \gamma}$ with $\gamma = O(1/\alpha)$~\cite{conf/crypto/LyubashevskyM09,conf/icalp/BaiSW16}, so sufficiently strong (fine-grained) hardness of the former problem translates to corresponding hardness for the latter problem. A similar reduction holds in reverse, but with a larger loss: $\alpha = \Omega(\sqrt{n/\log n}/\gamma)$.}
On the other hand, most concrete security estimates for lattice-based cryptosystems
are based on the runtimes of the fastest known (possibly heuristic) algorithms for exact or near-exact $\GapSVP$. %
So, apart from its inherent theoretical interest, understanding the fine-grained complexity of (near-)exact $\GapSVP$ and $\BDD$ sheds light on questions of great practical importance.

\paragraph{An addendum on exponential kissing number lattices.}
An earlier version of this work claimed that \cref{en:abs-improved-eth-kn-bdd,en:abs-improved-seth-bdd} in the abstract, corresponding to \cref{thm:gapeth-bdd-kn-informal,thm:gapseth-bdd-informal} below, held unconditionally. These claims were made based on the prior published works~\cite{Vladut19,Vladut-lpkissnum-2021}, which asserted that exponential kissing number lattices in the $\ell_2$ norm (and other $\ell_p$ norms) existed. I.e., using our notation,~\cite{Vladut19} asserted that $\kissnum > 0$.

Unfortunately,~\cite{bennett2024difficulties} showed that both of the works~\cite{Vladut19,Vladut-lpkissnum-2021} contained serious bugs, and both of those works were retracted. To the best of our knowledge, the core results in this work are (still) completely correct.
However, now the hardness results in \cref{thm:gapeth-bdd-kn-informal,thm:gapseth-bdd-informal} are conditional since they require that $\kissnum > 0$.
Additionally, some of the quantitative hardness results we derived as corollaries from the claimed explicit lower bound on $\kissnum$ in~\cite{Vladut19} are no longer known to hold.
(We emphasize that although the \emph{proofs} in~\cite{Vladut19,Vladut-lpkissnum-2021} were wrong, the \emph{claims made}, which are necessary to invoke our \cref{thm:gapeth-bdd-kn-informal,thm:gapseth-bdd-informal}, are plausibly---and perhaps even likely---true.)

\subsection{Our Results}

\begin{figure}[t]
    \centering
    \includegraphics[height=1.9in]{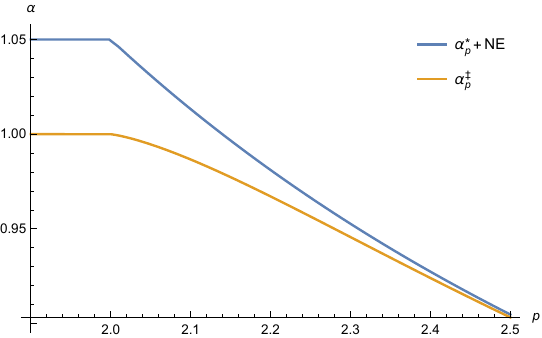}
    \caption{A plot showing the smallest relative distance~$\alpha$ for which $2^{\Omega(n)}$-hardness of $\BDD_{p, \alpha}$ is known under variants of ETH. $\alphaeth_p$ is from this work, and ``$\alpha^*_p$ + NE''---i.e., $\alpha_p^*$ with norm embeddings---is from~\cite{conf/coco/BennettP20}. A proof that $\kissnum > 0$ would immediately show additional improved hardness results for $\BDD_{p, \alpha}$ under variants of ETH.}
    \label{fig:alpha-plots}
\end{figure}

In this work, we show improved fine-grained hardness of $\BDD_{p, \alpha}$ and $\GapSVP_{p, \gamma}$, with an emphasis on results for smaller relative distance~$\alpha$ and larger approximation factor~$\gamma$, and on analyzing the complexity of the problems as the underlying~$\ell_p$ norm varies.
(We note that $\BDD_{p, \alpha'}$ trivially reduces to $\BDD_{p, \alpha}$ when $\alpha' < \alpha$, and so showing hardness results for $\BDD_{p, \alpha}$ for smaller $\alpha$ is showing something stronger.)

At a conceptual level, our work gives very general reductions to $\BDD$ (presented in \cref{thm:generic-bdd}), which reduce the task of showing hardness for $\BDD$ to analyzing properties of certain gadget lattices (described in \cref{subsubsec:intro-locally-dense}). The few known hardness results for $\BDD$ (essentially just \cite{conf/approx/LiuLM06,conf/coco/BennettP20} and this work) are all shown using this gadget lattice framework, but the previous works required separate reductions. The reductions in this work give a unified way to show hardness results using this framework.

At a technical level, our improved results for $\BDD$ follow from improved analysis of the techniques used in~\cite{conf/stoc/AggarwalS18} and~\cite{conf/coco/BennettP20} together with our new framework. 
Aggarwal and Stephens-Davidowitz~\cite{conf/stoc/AggarwalS18} presented three main results on the fine-grained hardness of $\GapSVP$, summarized in Items~1 to 3 in its abstract.
Bennett and Peikert~\cite{conf/coco/BennettP20} showed new hardness results for $\BDD$ by refining and adapting the analysis used to show~\cite{conf/stoc/AggarwalS18}, Item~1. Analogously, in this work we obtain two of our hardness results for $\BDD$ by refining and adapting the analysis used to show~\cite{conf/stoc/AggarwalS18}, Items~2 and 3. Specifically, \cref{thm:gapeth-bdd-kn-informal} corresponds to~\cite{conf/stoc/AggarwalS18}, Item~3 and \cref{thm:gapeth-bdd-int-informal} to~\cite{conf/stoc/AggarwalS18}, Item~2.
Our third hardness result for $\BDD$, presented in \cref{thm:gapseth-bdd-informal}, uses ideas from the other reductions together with our new framework.
Finally, our improved result for $\GapSVP$, presented in \cref{thm:gapseth-svp-informal}, answers a natural question left open by~\cite{conf/stoc/AggarwalS18,conf/soda/AggarwalBGS21}.

Our results assume (randomized or non-uniform) ``gap'' variants of the Exponential Time Hypothesis (ETH) and Strong Exponential Time Hypothesis (SETH). Recall that ``plain'' ETH asserts that solving the $3$-SAT problem on~$n$ variables requires $2^{\Omega(n)}$ time, and ``plain'' SETH asserts that for every $\eps > 0$ there exists $k \in \Z^+$ such that solving the $k$-SAT problem on $n$ variables requires $2^{(1-\eps)n}$-time.
The gap variants of these assumptions hypothesize that similar runtime lower bounds hold even for approximating the number of satisfiable clauses in a $k$-SAT formula to within some small constant approximation factor; see \cref{subsec:fine-grained-assumptions} for the precise definitions. We sometimes informally use the terminology ``(Gap-)ETH-hardness'' to denote $2^{\Omega(n)}$-hardness of a problem assuming a variant of (Gap-)ETH, and ``(Gap-)SETH-hardness'' to denote $2^{c n}$-hardness of a problem assuming a variant of (Gap-)SETH for some explicit constant $c > 0$.

\subsubsection{Hardness for BDD}

Our first result shows conditional improved exponential hardness of $\BDD_{p, \alpha}$ for all sufficiently small values of $p$, including the important Euclidean case of $p = 2$, assuming $\kissnum > 0$ and a variant of Gap-ETH. Indeed, assuming that $\kissnum > 0$, it improves the smallest value of $\alpha$ for which exponential hardness of $\BDD_{2, \alpha}$ is known under a general-purpose complexity-theoretic assumption to $\alpha < 2^{-\kissnum}$, showing such hardness for an
explicit\footnote{Each of the quantities $\alphakissnum$, $\alphaeth_p$, $\alphaseth_{p, C}$, $\alpha_{p, C}^*$, $\alpha_p^*$, and $C_p$ described in this section is ``explicit'' in the sense that it is expressible
via some (not necessarily closed-form) expression. These expressions are easily computed to high accuracy in practice as shown, e.g., in \cref{fig:alpha-plots}. Additionally, we emphasize that these quantities are constants in that they do not depend on the rank $n$ of the lattice in the corresponding problem.
}
constant less than the $\alpha = 1$ threshold for the first time.\footnote{Using the ideas in this paper and~\cite{conf/stoc/AggarwalS18}, showing exponential hardness of $\SVP$ essentially corresponds to showing exponential hardness of $\BDD$ with $\alpha = 1 - \eps$ for some constant $\eps > 0$. 
Additionally, using the $\BDD$-hardness framework in this paper, it would be relatively straightforward to show exponential hardness of $\BDD$ with $\alpha = 1 + \eps$ for any constant $\eps > 0$. So, the $\alpha = 1$ threshold is qualitatively quite natural, and trying to show hardness for an explicit constant $\alpha < 1$ is a natural goal.}

\begin{restatable}[Gap-ETH-hardness of BDD, first bound] {theorem}{gapethbddkn}\label{thm:gapeth-bdd-kn-informal}
    For all $p \in [1, \infty)$, there is no $2^{o(n)}$-time algorithm for $\BDD_{p, \alpha}$ for any constant $\alpha > \alphakissnum := 2^{-\kissnum}$, assuming that the $\ell_2$ kissing-number constant $\kissnum$ (defined in \cref{sec:gapeth-bdd-kn}) satisfies $\kissnum > 0$ and that non-uniform Gap-ETH holds.
\end{restatable}

Our second result shows improved exponential hardness of $\BDD_{p, \alpha}$ in a different regime and under a somewhat weaker assumption.

\begin{restatable}[Gap-ETH-hardness of BDD, second bound]{theorem}{gapethbddint} \label{thm:gapeth-bdd-int-informal}
    For all $p \in [1, \infty)$, there is no $2^{o(n)}$-time algorithm for $\BDD_{p, \alpha}$ for any constant $\alpha > \alphaeth_p$, assuming that randomized Gap-ETH holds. Here $\alphaeth_p$ is an explicit constant, defined in \cref{eq:alpha-eth}, which satisfies $\alphaeth_p = 1$ for $1 \leq p \leq 2$, $\alphaeth_p < 1$ for all $p > 2$, and $\alphaeth_p \to 1/2$ as $p \to \infty$.
\end{restatable}

In~\cite{conf/stoc/AggarwalS18}, Aggarwal and Stephens-Davidowitz showed SETH-hardness of $\GapSVP_{p, 1}$ for all $p > p_0 \approx 2.1397$.
To partially overcome the ``$p_0$ barrier,'' they generalized their proof techniques to show Gap-ETH-hardness of $\GapSVP_{p, \gamma}$ for all $p > 2$.
The results in \cite{conf/coco/BennettP20} adapted the former techniques of~\cite{conf/stoc/AggarwalS18} to show SETH-hardness of $\BDD$, and similarly got stuck at the $p_0$ barrier in the sense that they could not prove hardness of $\BDD_{p, \alpha}$ with $p < p_0$ and $\alpha < 1$ simultaneously.
The proof of \cref{thm:gapeth-bdd-int-informal} can be thought of as adapting the latter, generalized techniques of~\cite{conf/stoc/AggarwalS18} to $\BDD$, and analogously allows us to prove Gap-ETH-hardness of $\BDD_{p, \alpha}$ with $\alpha < 1$ for all $p > 2$.

Compared to related quantities, $\alphaeth_p$ is: at most ``$\alpha_p^*$ with norm embeddings'' for all $p \in [1, \infty)$; strictly less than $\alpha_p^*$ for all sufficiently small $p$; strictly less than $\alphakissnum$ for all sufficiently large $p$; and strictly less than the minimum of $\alpha_p^*$ and $\alphakissnum$ for intermediate values of $p$.
That is, $\alphaeth_p$ improves on both~$\alpha_p^*$ (even with norm embeddings) and~$\alphakissnum$ for intermediate values of~$p$; again, see the left plot in \cref{fig:alpha-plots}.
(Recall that~\cite{conf/coco/BennettP20} shows exponential hardness of $\BDD_{p, \alpha}$ for $\alpha > \alpha_p^*$ assuming randomized ETH, and \cref{thm:gapeth-bdd-kn-informal} above shows such hardness for $\alpha > \alphakissnum$ assuming non-uniform Gap-ETH.)
However, \cref{thm:gapeth-bdd-int-informal} relies on a somewhat stronger hardness assumption than the one used in~\cite{conf/coco/BennettP20}, and a somewhat weaker hardness assumption than \cref{thm:gapeth-bdd-kn-informal}, so the prior and new results are formally incomparable.

Our third result shows conditional $2^{n/C}$-hardness of $\BDD_{p, \alpha}$ for any $C > 1$ and $\alpha > \alphaseth_{p, C}$, where $\alphaseth_{p, C}$ is an explicit constant. As with \cref{thm:gapeth-bdd-kn-informal}, this result requires that $\kissnum > 0$.

\begin{restatable}[Gap-SETH-hardness of BDD]{theorem}{gapsethbdd}
    \label{thm:gapseth-bdd-informal}
    For all $p \in [1, \infty) \setminus 2 \Z$ and all $C > 1$, there is no $2^{n/C}$-time algorithm for $\BDD_{p, \alpha}$ for any constant $\alpha > \alphaseth_{p, C}$, assuming $\kissnum > 0$ and that non-uniform Gap-SETH holds. Here $\alphaseth_{p, C}$ is an explicit constant, defined in \cref{eq:alpha-seth} in terms of $\kissnum$, which satisfies $\alphaseth_{p, C} \to 1$ as $C \to \infty$ for any fixed $p \in [1, \infty)$.
\end{restatable}

We again note that \cref{thm:gapseth-bdd-informal,thm:gapeth-bdd-kn-informal} were claimed unconditionally in a prior version of this work based on the prior published work~\cite{Vladut19}, which was subsequently retracted. We also remark that assuming non-uniform variants of Gap-(S)ETH in \cref{thm:gapseth-bdd-informal,thm:gapeth-bdd-kn-informal} is essentially only necessary because it is not clear that (bases of) exponential kissing number lattices can be computed efficiently, even if they exist. If (bases of) such lattices not only existed but were in fact efficiently computable, then it would likely be possible to reprove our theorem statements under weaker \emph{randomized} Gap-(S)ETH assumptions (as opposed to \emph{non-uniform} Gap-(S)ETH assumptions).

\subsubsection{Hardness for SVP}

Our final result shows the same strong runtime lower bounds for $\GapSVP_{p, \gamma}$ with some constant $\gamma > 1$ under (randomized) Gap-SETH as~\cite{conf/stoc/AggarwalS18} showed for $\GapSVP_{p, 1}$ under (randomized) SETH.
 This answers a natural question left open by~\cite{conf/soda/AggarwalBGS21}, which analogously showed the same runtime lower bounds for $\GapCVP_{p, \gamma}$ with some constant $\gamma > 1$ under Gap-SETH as~\cite{conf/focs/BennettGS17} showed for the Closest Vector Problem ($\GapCVP$) under SETH.

\begin{restatable}[Gap-SETH-hardness of SVP]{theorem}{gapsethsvp} \label{thm:gapseth-svp-informal}
   For all $p > p_0 \approx 2.1397$, $p \notin 2\Z$, and all $C > C_p$, there is no $2^{n/C}$-time algorithm for $\GapSVP_{p, \gamma}$ for some constant $\gamma > 1$, assuming that randomized Gap-SETH holds. Here $C_p > 1$ is an explicit constant, defined in \cref{eq:Cp-def}, which satisfies $C_p \to 1$ as $p \to \infty$.
\end{restatable}

The reduction used to prove \cref{thm:gapseth-svp-informal} is itself a natural modification of the reduction used in~\cite{conf/stoc/AggarwalS18} to prove SETH-hardness of exact $\GapSVP$, but its analysis is more nuanced. We emphasize that simply plugging an instance of $\binGapCVP_{p, \gamma}$ with $\gamma > 1$ rather than $\gamma = 1$ into the reduction of~\cite{conf/stoc/AggarwalS18} does not yield corresponding hardness of approximation for $\GapSVP_{p, \gamma'}$ with $\gamma' > 1$; a modified reduction is necessary. Finally, we remark that the somewhat odd-looking $p \notin 2\Z$ requirement in \cref{thm:gapseth-bdd-informal,thm:gapseth-svp-informal} is an artifact of the ``upstream'' hardness results we employ for $\binGapCVP_{p, \gamma}$; see \cref{thm:gapseth-cvp}.

\subsection{Our Techniques}
\subsubsection{Locally Dense Lattices}
\label{subsubsec:intro-locally-dense}

As in nearly all prior work on the complexity of $\BDD$ and $\GapSVP$ (e.g.,~\cite{journals/tit/Micciancio01,journals/jacm/Khot05,conf/approx/LiuLM06,conf/stoc/AggarwalS18,conf/coco/BennettP20}), a key component of our results is the construction of a family of ``locally dense'' lattices, which are specified by a lattice $\lat^{\dagger}$ and corresponding target vector $\vec{t}^{\dagger}$. For our purposes, a locally dense lattice $\lat^{\dagger}$ is one with few ``short'' vectors, many vectors ``close'' to $\vec{t}^{\dagger}$, but few vectors ``too close'' to $\vec{t}^{\dagger}$.
(Other works such as~\cite{conf/coco/Micciancio14} define locally dense lattices in a closely related but different way, e.g., without the requirement of few ``too close'' vectors.)

For a discrete set $S$, which we will take to be a lattice or a subset of a lattice, define 
\begin{align*}
    N_p(S, r, \vec{t}) &:= \abs{\set{\vec{x} \in S : \norm{\vec{t} - \vec{x}}_p \le r}}
    \ \text, \\
    N_p^o(S, r, \vec{t}) &:= \abs{\set{\vec{x} \in S : \norm{\vec{t} - \vec{x}}_p < r}}
    \ \text.
\end{align*}
Somewhat more formally, we define a locally dense lattice~$\latdag$, $\tdag$ with relative distance~$\alpha_G$ in the~$\ell_p$ norm to be one for which
\begin{equation}
    N_p(\latdag, \alpha_G, \tdag) \geq \nu^n \cdot N_p^o(\latdag, 1, \vec{0})
    \label{eq:close-vs-short-intro}
\end{equation} for some $\nu > 1$. That is, $\latdag$, $\tdag$ is such that the number $G$ of ``close vectors'' (within distance $\alpha_G$ of $\tdag$) is an exponential factor larger than the number of short vectors (of norm at most one). Similarly, we require there to be an exponential factor more close vectors than ``too close'' vectors, along with some other technical conditions. We defer discussing these issues until the main body of the paper, and for the remainder of the introduction focus on the constraint in \cref{eq:close-vs-short-intro}.

A crux in showing hardness of 
$\BDD_{p, \alpha}$ and $\GapSVP_{p, \gamma}$ is constructing good locally dense lattices, and their parameters govern the precise hardness results that we can obtain. A family of locally dense lattices with smaller relative distance~$\alpha_G$
and larger $\nu$ leads to stronger hardness results.
To obtain ETH-type hardness results, we simply need $\nu$ to be a constant greater than $1$, and then we can show $2^{\Omega(n)}$-hardness of $\BDD_{p, \alpha}$ for any constant $\alpha > \alpha_G$.
For SETH-type hardness results, we get $2^{n/C}$-hardness of $\BDD_{p, \alpha}$ whenever our reduction to $\BDD$ has a multiplicative rank-increase factor of~$C$.
The value of $C$ depends on the gap factor $\nu$ in \cref{eq:close-vs-short-intro}, so to show such hardness for explicit $C > 0$ we need an explicit lower bound on $\nu$.
Our reductions also give a tradeoff between~$C$ and~$\alpha$, as shown in the right plot in \cref{fig:alpha-plots}.
The full situation is actually a bit more complicated when taking ``too close'' vectors into account, but we again defer discussing this for now.
The situation for $\GapSVP$ is similar to the situation for $\BDD$.

\subsubsection{Sparsification}

An important technique in our work is randomized lattice sparsification, an efficient algorithm that essentially does the following.
Given (a basis of) a lattice $\lat$ and an index $q \in \Z^+$ as input, the algorithm randomly samples a sublattice $\lat' \subseteq \lat$ such that for any fixed, finite set of lattice points $S \subseteq \lat$ satisfying some mild conditions, $\card{S \cap \lat'} \approx \card{S}/q$ with probability near $1$.
A variant of this algorithm, additionally given $\vec{t} \in \lspan(\lat)$ as input, randomly samples $\lat' \subseteq \lat$ and~$\vec{t}'$ such that for any fixed, finite set of points $S \subseteq \lat - \vec{t}$ satisfying some mild conditions, $\card{S \cap (\lat' - \vec{t}')} \approx \card{S}/q$ with probability near $1$.

Intuitively, some mild caveats aside, sparsification says that a lattice with few short vectors (and few ``too close'' vectors) is just as good as a lattice with \emph{no} short non-zero vectors (and no ``too close'' vectors), since the latter can be efficiently obtained from the former.
Indeed, sparsifying with index $q \approx N_p^o(\lat, r, \vec{0})$ allows us to turn a lattice~$\lat$ and target~$\vec{t}$ satisfying, say, $N_p(\lat, \alpha r, \vec{t}) \geq 100 \cdot N_p^o(\lat, r, \vec{0})$ into a lattice~$\lat'$ and target~$\vec{t}'$ with $N_p(\lat', \alpha r, \vec{t}') \geq 1$ and $N_p^o(\lat' \setminus \set{\vec{0}}, r, \vec{0}) = 0$, so $\dist_p(\vec{t}', \lat) \leq \alpha r$ and $\sm_1(\lat') \geq r$.
That is, the output $\lat'$, $\vec{t}'$ satisfies the $\BDD$ promise $\dist_p(\vec{t}', \lat) \leq \alpha \cdot \sm_1(\lat')$.
See \cref{subsec:sparsification} for a formal description of sparsification.

\subsubsection{A Transformation Using Locally Dense Lattices}

Define 
$\binGapCVP_{p, \gamma}$ to be the following variant of the decision version of the $\gamma$-approximate Closest Vector Problem: given a basis~$B$ of a rank-$n$ lattice~$\lat$ and a target vector~$\vec{t}$ as input, decide whether there exists $\vec{x} \in \bit^n$ such that $\norm{B\vec{x} - \vec{t}}_p \leq 1$ (i.e., the input is a YES instance), or whether $\dist_p(\vec{t}, \lat) > \gamma$ (i.e., the input is a NO instance), under the promise that one of the two cases holds. In other words, $\binGapCVP$ is the variant of $\GapCVP$ that asks whether there is a \emph{binary} combination of basis vectors close to the target. Much is known about the (fine-grained) complexity of $\binGapCVP$, 
which will be useful for us (see \cref{thm:gapeth-cvp,thm:gapseth-cvp}).

Our reductions from $\binGapCVP$ to $\BDD$ and to $\GapSVP$ have the same basic form. Given a rank-$n'$ instance $B',\vec{t}'$ of $\binGapCVP_{p, \gamma}$ for some $\gamma > 1$, we apply the following transformation with some scaling factors $s, \ell > 0$ and some locally dense lattice $\latdag = \lat(\Bdag)$ of rank $n - n'$ with target $\tdag$ satisfying \cref{eq:close-vs-short-intro}:
\begin{equation}
    \mat{B} := \begin{pmatrix} s \mat{B}' & \mat{0} \\ \mat{I}_{n'} & \mat{0} \\ \mat{0} & \ell \mat{B}^\dagger \end{pmatrix}
    \ \text, \qquad
    \vec{t} := \begin{pmatrix} s \vec{t}' \\ \tfrac{1}{2} \vec{1}_{n'} \\ \ell \vec{t}^\dagger \end{pmatrix}
    \ \text.
\label{eq:transformation-intro}
\end{equation}
Essentially the same transformation appears in both~\cite{conf/stoc/AggarwalS18} and~\cite{conf/coco/BennettP20}, and similar ideas appear in a number of works before that. Our work differs in its constructions of locally dense lattice gadgets $(\latdag$, $\tdag)$, its more general reductions, and its improved analysis.

Here we give a rough analysis of the transformation using two observations.
First, we observe that appending $I_{n'}$ to $B'$ allows us to upper bound the number of short lattice vectors
in $\lat(B)$ by
\begin{equation}
    N_p^o(\lat(B), r', \vec{0}) \leq N_p^o(\Z^{n'} \oplus \lat(\ell B^{\dagger}), r', \vec{0})
    \label{eq:intro-short-vec-ub}
\end{equation}
for any $r' > 0$.
Second, suppose that $B'$, $\vec{t}'$ is a YES instance of $\binGapCVP$.
Then there exists $\vec{x} \in \bit^{n'}$ such that $\norm{B' \vec{x} - \vec{t}'} \leq 1$, and hence for each $\vec{y} \in \Z^{n-n'}$ with $\norm{B^{\dagger}\vec{y} - \vec{t}^{\dagger}}_p \leq \alpha_G$ we get that $\norm{B(\vec{x}, \vec{y}) - \vec{t}}_p \leq r$, where $r := (s^p + n'/2^p + (\alpha_G \cdot \ell)^p)^{1/p}$.
So,
\begin{equation}
    N_p(\lat(B), r, \vec{t}) \geq N_p(\lat(B^{\dagger}), \alpha_G, \vec{t}^{\dagger}) \ \text.
    \label{eq:intro-close-vec-lb}
\end{equation}

To transform a YES instance of $\binGapCVP_{p, \gamma}$ to a valid instance of $\BDD_{p, \alpha}$ for some $\alpha > 0$, it essentially suffices to set the parameters $r, s, \ell$ and use suitable $\Bdag$, $\tdag$ so that, say,
\begin{equation}
N_p(\lat(B^{\dagger}), \alpha_G, \vec{t}^{\dagger}) \geq 100 \cdot N_p^o(\Z^{n'} \oplus \lat(\ell B^{\dagger}), r/\alpha, \vec{0}) \ \text.
\label{eq:lld-goal-intro}
\end{equation}
Indeed, if \cref{eq:lld-goal-intro} holds, then by \cref{eq:intro-short-vec-ub,eq:intro-close-vec-lb}, $N_p(\lat(B), r, \vec{t}) \geq 100 \cdot N_p^o(\lat(B), r/\alpha, \vec{0})$.
We can then sparsify $\lat(B)$ to obtain a lattice with no non-zero vectors of norm less than~$r/\alpha$, and at least one vector within distance $r$ of~$\vec{t}$, as needed.

We recall that by assumption, $N_p(\lat(B^{\dagger}), \alpha_G, \vec{t}^{\dagger}) \geq \nu^{n-n'}$, which is important for satisfying \cref{eq:lld-goal-intro} since $N_p^o(\Z^{n'} \oplus \lat(\ell B^{\dagger}), r/\alpha, \vec{0})$ will typically be exponentially large in $n'$.
We also need that if the input $\binGapCVP$ instance is a NO instance, then there will be few vectors in~$\lat(B)$ that are close to~$\vec{t}$, which depends on~$\latdag$ having few vectors ``too close'' to $\tdag$, but again we defer discussing this. See \cref{lem:reduction-gadget} for a precise description of the useful properties of the transformation given in \cref{eq:transformation-intro}.

When reducing to $\GapSVP_{p, \gamma'}$ instead of $\BDD_{p, \alpha}$ we apply a further transformation to $B, \vec{t}$ before sparsifying. Namely, we apply Kannan's embedding, which appends the vector $(\vec{t}, u)$, for some value $u > 0$, to~$B$ to obtain a new basis:
\[
B, \vec{t}, u \mapsto \begin{pmatrix} B & \vec{t} \\ 0 & u \end{pmatrix} \ \text.
\]
The analysis in this case is a bit more subtle as well---we need to upper bound quantities of the form $N_p(\lat(B), (r^p - (w u)^p)^{1/p}, w \cdot (\vec{t}, u))$ not just for $w = 0, 1$ (corresponding to short and ``too close'' vectors in the $\BDD$ case, respectively) but for all integers $w \geq 2$ too---but the idea is similar. 
In fact, we use a result from~\cite{conf/stoc/AggarwalS18} (presented in \cref{thm:agcvp-to-svp}) that analyzes the combination of Kannan's embedding and sparsification already, essentially reducing our task to bounding the quantities $N_p(\lat(B), (r^p - (w u)^p)^{1/p}, w \cdot (\vec{t}, u))$.

\subsubsection{Specific Locally Dense Lattices}
We conclude this summary of techniques by describing the specific locally dense lattices $\lat^{\dagger}, \vec{t}^\dagger$ 
that we use to instantiate \cref{eq:transformation-intro}.
We use two main families of locally dense lattices for our results. 

\paragraph{Exponential kissing number lattices.}
The first family of locally dense lattices is derived from a family of ``exponential kissing number'' lattices $\set{\lat_n}_{n=1}^{\infty}$, assuming such a family of lattices exists. More precisely, we call $\set{\lat_n}_{n=1}^{\infty}$ a family of exponential kissing number lattices if for every $n \in \Z^+$, $\lat_n$ is of rank~$n$ and has exponential Euclidean kissing number, i.e., $N_2(\lat_n, \lambda_1(\lat_n), \vec{0}) = 2^{\kissnum n - o(n)}$ for some constant $\kissnum > 0$.

We recall that it is not currently known whether $\kissnum > 0$ (this was claimed in~\cite{Vladut19}, however that work contained a bug).
Previously,~\cite{conf/stoc/AggarwalS18} showed how to use the existence of such a family to prove $2^{\Omega(n)}$-hardness of $\GapSVP_{p, \gamma}$ for all $p \geq 1$ and some $\gamma > 1$ (and in particular, for $1 \leq p \leq 2$, for which the result was not already known from other techniques), assuming non-uniform Gap-ETH.

The proofs of \cref{thm:gapeth-bdd-kn-informal,thm:gapseth-bdd-informal} both use a family of exponential kissing number lattices (assuming that such a family exists) to construct locally dense lattices, but in different ways.
The proof of \cref{thm:gapeth-bdd-kn-informal} constructs a locally dense lattice $\latdag$, $\tdag$ with relative distance $\alpha \approx 2^{-\kissnum}$, but with a non-explicit lower bound on the gap factor $\nu$ in \cref{eq:close-vs-short-intro}. The proof of \cref{thm:gapseth-bdd-informal} constructs a locally dense lattice $\latdag$, $\tdag$ with relative distance $\alpha \approx 1$ but with an explicit lower bound on $\nu$---essentially $\nu \geq 2^{\kissnum}$. The values $\alphaseth_{p, C}$ in \cref{thm:gapseth-bdd-informal} are defined to be a certain quantity relating the maximum possible kissing number in a lattice of rank $(C - 1)n'$, roughly $2^{\kissnum \cdot (C - 1)n'}$, and the number of vectors in $\Z^{n'}$ of norm at most $r$ for some $r > 0$.

The proof of \cref{thm:gapseth-bdd-informal} is actually somewhat simpler than that of \cref{thm:gapeth-bdd-kn-informal}, so we first give a bit more detail on it. We note that taking $\lat^{\dagger} := \lat_{n}$ and $\vec{t}^{\dagger} := \vec{0}$ for a exponential kissing number lattice $\lat_{n}$ \emph{almost} yields a locally dense lattice family with $\nu^{n-o(n)}$ many close vectors for $\nu \geq 2^{\kissnum}$ and relative distance $\alpha = 1$, but there are two issues:
(1) We assume the existence of exponential kissing number lattices with respect to the~$\ell_2$ norm rather than general~$\ell_p$ norms, and
(2) the origin is itself a ``too close'' vector.

We handle issue~(1) by applying norm embeddings~\cite{conf/stoc/RegevR06} with distortion $(1 + \eps)$ to $\set{\lat_n}_{n=1}^{\infty}$ to obtain a family of lattices $\set{\lat'_n}_{n=1}^{\infty}$ with exponentially high ``handshake number'' in the $\ell_p$ norm: \[ N_p(\lat'_n, (1 + \eps) \cdot \sm_1(\lat'_n), \vec{0}) \geq 2^{\kissnum n - o(n)} \ \text, \] where applying the norm embedding is efficient for any $\eps \geq 1/\poly(n)$. 
We handle issue~(2) by ``sparsifying away the origin.'' 
Specifically, for all sufficiently large $n$ we show how to sample a
sublattice $\lat_n'' \subseteq \lat_n'$ and $\vec{t}'' \in \lat' \setminus \lat''$ satisfying \[ N_p(\lat''_n, (1 + \eps) \cdot \sm_1(\lat''_n), \vec{t}'') \geq N_p(\lat'_n, (1 + \eps) \cdot \sm_1(\lat''_n), \vec{0})/4 \geq 2^{\kissnum n - o(n)} \] with positive probability. In particular, this shows that such lattices exist.

For the proof of \cref{thm:gapeth-bdd-kn-informal}, we take $\lat^{\dagger}$ to be $\lat_{n}$ scaled so that $\lambda_1(\lat_n) = 1$, and take $\vec{t}^{\dagger}$ to be a uniformly random vector of norm $\delta$ for some appropriately chosen constant $0 < \delta < 1$. We then analyze $N_2(\latdag, (1 - \eps) \cdot \lambda_1(\latdag), \tdag)$ for some appropriately chosen constant $0 < \eps < \delta$. Intuitively, there is a tradeoff between choosing smaller $\delta$, which makes $N_2(\latdag, (1 - \eps) \cdot \lambda_1(\latdag), \tdag)$ larger but requires $\eps$ to be smaller, and larger $\delta$, which makes $N_2(\latdag, (1 - \eps) \cdot \lambda_1(\latdag), \tdag)$ smaller but allows for $\eps$ to be larger. The relative distance $\alpha$ that our reduction achieves is essentially the smallest $\alpha = 1 - \eps$ for which we can ensure that $N_2(\latdag, (1 - \eps) \cdot \lambda_1(\latdag), \tdag) \geq 2^{\Omega(n)}$.
To translate these results to general $\ell_p$ norms, we again use norm embeddings. We also use additional techniques for dealing with ``too close'' vectors.

We note that the construction of locally dense lattices from lattices with exponential kissing number in~\cite{conf/stoc/AggarwalS18} does not give explicit bounds on the relative distance $\alpha$ or gap factor $\nu$ achieved; the reduction there essentially just needs $\nu > 1$ and $\alpha = 1 - \eps$ for non-explicit $\eps > 0$.
On the other hand, the construction in \cref{thm:gapeth-bdd-kn-informal} gives an explicit bound on $\alpha$ but not on $\nu$, and the construction in \cref{thm:gapseth-bdd-informal} gives an explicit bound on $\nu$ but with $\alpha > 1$. Therefore, \cref{thm:gapeth-bdd-kn-informal,thm:gapseth-bdd-informal} can be seen as different refinements of the corresponding analysis in~\cite{conf/stoc/AggarwalS18}.
A very interesting question is whether it's possible to get a construction that simultaneously achieves explicit $\nu$ and relative distance $\alpha = 1 - \eps$; our current techniques do not seem to be able to achieve this. Such a construction would lead to new (Gap-)SETH-hardness results for $\GapSVP$.

\paragraph{The integer lattice $\Z^n$.}
The second family of locally dense lattices that we consider, used to prove \cref{thm:gapeth-bdd-int-informal,thm:gapseth-svp-informal}, simply takes $\lat^{\dagger}$ to be the integer lattice $\Z^n$, and $\vec{t}^{\dagger}$ to be the all-$t$s vector for some constant~$t$ (without loss of generality, $t \in [0, 1/2]$):
\[
 \latdag := \Z^n \ \text, \qquad \tdag := t \cdot \vec{1}_n \ \text.
\]
This family was also used in~\cite{conf/stoc/AggarwalS18} and~\cite{conf/coco/BennettP20}.

We will be especially interested in the case where $t = 1/2$. In this case $N_p(\Z^n, \dist_p(1/2 \cdot \vec{1}_n, \Z^n), 1/2 \cdot \vec{1}_n) = 2^n$, where $\dist_p(1/2 \cdot \vec{1}_n, \Z^n) = n^{1/p}/2$. So, for our analysis it essentially suffices to upper bound $N_p(\Z^n, r, \vec{0})$ for some $r = n^{1/p}/(2 \alpha)$.
We have good techniques for doing this; see \cref{subsec:point-counting-Zn}.
(We note that the ``$p_0$ barrier'' mentioned earlier comes from $p_0$ being the smallest value of $p$ satisfying $N_p(\Z^n, n^{1/p}/2, \vec{0}) \leq 2^n$.) %

The SETH-hardness result for $\GapSVP_{1, p}$ for $p > p_0$ in~\cite{conf/stoc/AggarwalS18}, the hardness results for $\BDD$ in~\cite{conf/coco/BennettP20}, and the Gap-SETH-hardness result for $\GapSVP_{\gamma, p}$ in \cref{thm:gapseth-svp-informal} all take $\tdag = 1/2 \cdot \vec{1}$, and analyze $N_p(\Z^n, n^{1/p}/2, 1/2 \cdot \vec{1})/N_p(\Z^n, n^{1/p}/(2\alpha), \vec{0}) = 2^n/N_p(\Z^n, n^{1/p}/(2\alpha), \vec{0})$ for some $\alpha > 0$.
That is, they essentially just need to upper bound $N_p(\Z^n, n^{1/p}/(2\alpha), \vec{0})$ for some $\alpha$ (with $\alpha = 1$ for the $\GapSVP$ hardness results).
As alluded to in the discussion after the statement of \cref{thm:gapeth-bdd-int-informal}, the Gap-ETH-hardness result in~\cite{conf/stoc/AggarwalS18} for $\GapSVP_{p, \gamma}$ essentially works by proving that for every $p > 2$ there exist $t \in (0, 1/2]$ and $r > 0$ (not necessarily $t = 1/2$ or $r = n^{1/p}/2$) such that
\begin{equation}
\frac{N_p(\Z^n, r, t \cdot \vec{1})}{N_p(\Z^n, r/\alpha, \vec{0})} \geq 2^{\Omega(n)}
\label{eq:intro-gapeth-opt}
\end{equation} 
for some non-explicit $\alpha < 1$.

We do something similar, but study the more refined question of what the \emph{minimum} value of $\alpha$ is such that \cref{eq:intro-gapeth-opt} holds for some $t$ and $r$.
This minimum value of $\alpha$ is essentially how we define the quantities $\alphaeth_p$ used in \cref{thm:gapeth-bdd-int-informal}; see \cref{eq:alpha-eth} for a precise definition.
We note that, interestingly, in experiments this minimum value of $\alpha$ is always achieved by simply taking $t = 1/2$. That is, empirically it seems that we do not lose anything by fixing $t = 1/2$ and only varying $r$.\footnote{This is especially interesting since~\cite{elkies91:_packing_densities} notes that $N_p(\Z^n, r, t \cdot \vec{1})$ is \emph{not} maximized by $t = 1/2$ for some $p > 2$, including $p = 3$, and some \emph{fixed} $r > 0$. For $1 \leq p \leq 2$,
\cite{mazo90:_lattice_points} and~\cite{elkies91:_packing_densities} note that for any fixed $r > 0$, $N_p(\Z^n, r, t \cdot \vec{1})$ is in fact \emph{minimized} (up to a subexponential error term) by taking $t = 1/2$.}
We leave proving this as an interesting open question, but note that the strength of our results does not depend on its resolution either way.

\subsection{Open Questions}
One of the most interesting aspects of this and other work on the complexity of lattice problems is the interplay between geometric objects---here, lattices with exponential kissing number and locally dense lattices generally---and hardness results. 
Proving a better lower bound on~$\kissnum$ would immediately translate into an improved bound on the values of $\alphakissnum$ and $\alphaseth_{p, C}$, and more generally proving the existence of some family of gadgets with smaller relative distance~$\alpha$ and at least $2^{\Omega(n)}$ close vectors would translate into a hardness result improving on both \cref{thm:gapeth-bdd-kn-informal,thm:gapeth-bdd-int-informal}.

There is also the question of \emph{constructing} locally dense lattices.
The difference between existence and efficient randomized construction of locally dense lattices roughly corresponds to needing ``non-uniform'' versus ``randomized'' hardness assumptions. %
It is also an interesting question whether randomness is needed at all for showing hardness of $\BDD$ or $\GapSVP$. In this work we crucially use randomness for sparsification in addition to using it to construct locally dense lattices. Indeed, derandomizing hardness reductions for 
$\GapSVP$ (and similarly, $\BDD$) is a notorious, decades-old open problem.

\subsection{Acknowledgments}
We thank Noah Stephens-Davidowitz for helpful comments.

\section{Preliminaries}

We denote column vectors by boldface, lowercase letters (e.g., $\vec{u}, \vec{v}, \vec{w}$). We occasionally abuse notation and write things like $\vec{w} = (\vec{u}, \vec{v})$ instead of $\vec{w} = (\vec{u}^\top, \vec{v}^\top)^\top$.
We use $\vec{0}_n$ and $\vec{1}_n$ to denote the all-$0$s and all-$1$s vectors of dimension~$n$, respectively. We sometimes omit the subscript $n$ when it is clear from context.
Finally, we occasionally abuse notation by mixing elements of a finite field~$\F_q$ (for prime~$q$) and integers when performing arithmetic. In this case, we are really associating elements of $\F_q$ with some distinguished set of integer representatives, e.g., $\set{0, 1, \ldots, q - 1}$.

\subsection{Lattices and Point Counting}

A \emph{lattice} $\lat$ is a discrete additive subgroup of $\R^m$;
concretely, a lattice is the set of all integer linear combinations of a set of linearly independent vectors in $\R^m$.
This set of vectors is called a \emph{basis} of the lattice. Its cardinality $n$ is defined to be the \emph{rank} of the lattice (which turns out to be independent of the choice of basis), and the dimension $m$ of the basis vectors is defined to be the \emph{dimension} of the lattice.
A basis is often represented by a matrix $\mat{B}$ whose columns are the vectors in the basis, and the lattice generated by $\mat{B}$ is denoted $\lat(\mat{B})$.
Using this representation, a rank-$n$ lattice $\lat \subset \R^m$ with basis $\mat{B} \in \R^{m \times n}$ can be written as $\lat = \lat(B) = \mat{B} \cdot \Z^n$.
We denote the Moore-Penrose pseudo-inverse of a matrix $B \in \R^{m \times n}$ by $B^+ := (B^\top B)^{-1} B^\top$. When $B$ is a basis and $\vec{v} \in \lat(B)$ is a lattice vector, $B^+ \vec{v} = \vec{a} \in \Z^n$ is its coefficient vector.

We recall that for $p \in [1, \infty)$ the $\ell_p$ norm of a vector $\vec{x} \in \R^m$ is defined as $\norm{\vec{x}}_p := \bigl( \sum_{i = 1}^m \abs{x_i}^p \bigr)^{1/p}$,
and for $p = \infty$ it is defined as $\norm{\vec{x}}_\infty := \max_{i = 1}^m \abs{x_i}$.
The \emph{minimum distance} $\sm_1(\lat)$ of a lattice $\lat$ in the $\ell_p$ norm is defined to be the minimum length of nonzero vectors in the lattice, i.e.,
\begin{equation*}
    \sm_1(\lat) := \min_{\vec{v} \in \lat \setminus \set{\vec{0}}} \norm{\vec{v}}_p
    \ \text.
\end{equation*}
Equivalently, as the name suggests, it is the minimum distance between any two distinct vectors in the lattice.
For $p = 2$, we simply write $\lambda_1$ for $\sm[2]_1$.
We denote the distance between a vector $\vec{t}$ and a lattice $\lat$ in the $\ell_p$ norm by
\[
    \dist_p(\vec{t}, \lat) := \min_{\vec{x} \in \lat} \norm{\vec{t} - \vec{x}}_p
    \ \text.
\]

For $p \in [1, \infty]$, a discrete set $S \subset \R^m$, a target $\vec{t} \in \R^m$, and a distance $r \ge 0$, we define the following two point-counting functions:
\begin{align*}
    N_p(S, r, \vec{t}) &:= \abs{\set{\vec{x} \in S : \norm{\vec{t} - \vec{x}}_p \le r}}
    \ \text, \\
    N_p^o(S, r, \vec{t}) &:= \abs{\set{\vec{x} \in S : \norm{\vec{t} - \vec{x}}_p < r}}
    \ \text.
\end{align*}
We will typically use a lattice or a subset of a lattice as the discrete set $S$.

In the the following claim we observe two simple bounds on point counts in lattices.

\begin{claim} \label{clm:count-properties}
For any lattice $\lat$, target $\vec{t} \in \lspan(\lat)$, and $r \ge 0$:
\begin{enumerate}
    \item \label{en:count-short-sm}
    $N_p^o(\lat, r, \vec{0}) \ge 2 r / \sm_1(\lat) - 1$;
    \item \label{en:count-close-triangle}
    $N_p(\lat, r, \vec{t}) \le N_p(\lat, 2 r, \vec{0})$.
\end{enumerate}
\end{claim}

\begin{proof}
For \cref{en:count-short-sm}, the inequality is trivial when $r = 0$. For $r > 0$, we have
\begin{equation*}
    N_p^o(\lat, r, \vec{0})
    \ge N_p^o(\lat(\vec{v}), r, \vec{0})
    = 1 + 2 \cdot \Bigl\lceil \frac{r}{\sm_1(\lat)} - 1 \Bigr\rceil
    \ge \frac{2 r}{\sm_1(\lat)} - 1
    \ \text,
\end{equation*}
where $\vec{v}$ is an arbitrary shortest vector in $\lat$ with $\norm{\vec{v}}_p = \sm_1(\lat)$.

For \cref{en:count-close-triangle}, if $N_p(\lat, r, \vec{t}) = 0$ then the inequality is trivial. Otherwise fix $\vec{v} \in \lat$ to be such that $\norm{\vec{t} - \vec{v}}_p \le r$. Then for every $\vec{u} \in \lat$ satisfying $\norm{\vec{t} - \vec{u}}_p \le r$, we have $\vec{u} - \vec{v} \in \lat$ and $\norm{\vec{u} - \vec{v}}_p \le 2 r$ by the triangle inequality.
Moreover, for $\vec{u}_1 \neq \vec{u}_2$, $\vec{u}_1 - \vec{v} \neq \vec{u}_2 - \vec{v}$.
Hence we know $N_p(\lat, r, \vec{t}) \le N_p(\lat, 2 r, \vec{0})$, as desired.
\end{proof}

Given lattice-target pairs $(\lat_1, \vec{t}_1)$ and $(\lat_2, \vec{t}_2)$, the following claim gives an upper bound on the number of close vectors in the ``direct sum lattice'' $\lat_1 \oplus \lat_2$ to $(\vec{t}_1, \vec{t}_2)$ in terms of the product of the numbers of close vectors in $\lat_1$ to $\vec{t}_1$ and in $\lat_2$ to $\vec{t}_2$.

\begin{claim}
\label{clm:count-direct-sum}
For lattices $\lat_1, \lat_2$, targets $\vec{t}_1 \in \lspan(\lat_1), \vec{t}_2 \in \lspan(\lat_2)$, and $r \ge \max \set{\dist_p(\vec{t}_1, \lat_1), \dist_p(\vec{t}_2, \lat_2)}$,
\begin{align*}
    N_p(\lat_1 \oplus \lat_2, r, (\vec{t}_1, \vec{t}_2)) &\le N_p(\lat_1, (r^p - \dist_p(\vec{t}_2, \lat_2)^p)^{1/p}, \vec{t}_1) \cdot N_p(\lat_2, (r^p - \dist_p(\vec{t}_1, \lat_1)^p)^{1/p}, \vec{t}_2)
    \ \text.
\end{align*}
Additionally, if $r < \dist_p(\vec{t}_1, \lat_1)$ or $r < \dist_p(\vec{t}_2, \lat_2)$, $N_p(\lat_1 \oplus \lat_2, r, (\vec{t}_1, \vec{t}_2)) = 0$.
Moreover, these results hold with $N^o_p$ in place of $N_p$.
\end{claim}

\begin{proof}
For every lattice vector $(\vec{v}_1, \vec{v}_2) \in \lat_1 \oplus \lat_2$, $\dist_p((\vec{t}_1, \vec{t}_2), (\vec{v}_1, \vec{v}_2)) \le r$ (respectively $< r$) holds only if $\norm{\vec{t}_1 - \vec{v}_1}_p^p + \dist_p(\vec{t}_2, \lat_2)^p \le r^p$ (resp.\ $< r$) and $\norm{\vec{t}_2 - \vec{v}_2}_p^p + \dist_p(\vec{t}_1, \lat_1)^p \le r^p$ (resp.\ $< r$). Hence the desired inequalities follow.
\end{proof}

\subsection{Sparsification}
\label{subsec:sparsification}

The following lemma shows that the number of vectors in a collection of $m$ pairwise linearly independent (respectively, distinct) vectors in $\F_q^n$ that satisfy a random linear constraint (respectively, random affine constraint) is relatively concentrated around its expectation $m/q$.
Similar results appear in~\cite{journals/jacm/Khot05} and~\cite{conf/soda/Stephens-Davidowitz16}.

\begin{lemma}
\label{lem:sparsification-lin-eqs}
Let $\delta > 0$, let $n \in \Z^+$, let $q$ be a prime, and let $\vec{a}_1, \ldots, \vec{a}_m \in \F_q^n$.
\begin{enumerate}
    \item
    \label{en:sparsification-linind}
    If $\vec{a}_i \neq s \vec{a}_j$ for all $i \neq j$ and $s \in \F_q$, then
    \begin{equation*}
        \Pr_{\vec{x} \sim \F_q^n}[\card{\set{i \in [m] : \iprod{\vec{a}_i, \vec{x}} = 0}} \leq (1 - \delta) m/q] \leq \frac{q}{\delta^2 m}
        \ \text.
    \end{equation*}
    
    \item 
    \label{en:sparsification-distinct}
    If $\vec{a}_i \neq \vec{a}_j$ for all $i \neq j$, then
    \begin{equation*}
        \Pr_{\substack{\vec{x} \sim \F_q^n, \\ y \sim \F_q}}[\card{\set{i \in [m] : \iprod{\vec{a}_i, \vec{x}} = y}} \leq (1 - \delta) m/q] \leq \frac{q}{\delta^2 m}
        \ \text.
    \end{equation*}
\end{enumerate}
In particular, for $\delta = 1$, the probabilities are upper bounded by $q / m$.
\end{lemma}

\begin{proof}
We start by proving \cref{en:sparsification-linind}.
Let $X_i$ be an indicator random variable for the event $\iprod{\vec{a}_i, \vec{x}} = 0$, where $\vec{x} \sim \F_q^n$.
It holds that $\E[X_i] = 1/q$, that $\Var[X_i] = 1/q \cdot (1 - 1/q) \leq 1/q$, and that $X_1, \ldots, X_m$ are pairwise independent. 
For the last point, because $\vec{a}_i$ and $\vec{a}_j$ for $i \neq j$ are linearly independent, the linear map $\vec{x} \mapsto (\iprod{\vec{a}_i, \vec{x}}, \iprod{\vec{a}_j, \vec{x}})$ has a kernel of dimension $n - 2$. It follows that $\iprod{\vec{a}_i, \vec{x}} = \iprod{\vec{a}_j, \vec{x}} = 0$ for a $q^{n-2}/q^n = 1/q^2$ fraction of vectors $\vec{x} \in \F_q^n$, and hence that $X_i, X_j$ are pairwise independent.
The result then follows by applying Chebyshev's inequality to $\sum_{i=1}^m X_i$.

We next show that \cref{en:sparsification-linind} implies \cref{en:sparsification-distinct}. 
For each $i \in [m]$, let $\vec{a}_i' := (\vec{a}_i, -1)$, and note that if $\vec{a}_i \neq \vec{a}_j$ then $\vec{a}_i' \neq s \vec{a}_j'$ for all $s \in \F_q$. 
Moreover, for any $\vec{x} \in \F_q^n, y \in \F_q$, $\iprod{\vec{a}_i, \vec{x}} = y$ if and only if $\iprod{\vec{a}_i', \vec{x}'} = 0$, where $\vec{x}' = (\vec{x}, y)$.
Therefore by \cref{en:sparsification-linind},
\begin{equation*}
    \Pr_{\substack{\vec{x} \sim \F_q^n, \\ y \sim \F_q}}[\card{\set{i : \iprod{\vec{a}_i, \vec{x}} = y}} \leq (1 - \delta) m/q]
    = \Pr_{\substack{\vec{x}' \sim \F_q^{n+1}}}[ \lvert \set{i : \iprod{\vec{a}_i', \vec{x}'} = 0 } \rvert \leq (1 - \delta) m/q] %
    \leq \frac{q}{\delta^2 m}
    \ \text.
    \qedhere
\end{equation*}
\end{proof}

The following lemma shows that when $q \gg m$, it is probable that no vector in the collection satisfies a random affine constraint.
This is similar to \cite[Lemma 2.8]{conf/coco/BennettP20}.

\begin{lemma}
\label{lem:sparsification-upper-bound}
Let $n \in \Z^+$, let $q$ be a prime, and let $\vec{a}_1, \ldots, \vec{a}_m \in \F_q^n$. Then:
\begin{enumerate}
    \item
    \label{en:sparsification-ub-random-y}
    It holds that
    \[
    \Pr_{\substack{\vec{x} \sim \F_{q}^{n} \\ y \sim \F_q}}[\exists\ i \in [m] \text{ such that } \iprod{\vec{a}_i, \vec{x}} = y] \leq \frac{m}{q}
    \ \text.
    \]
    \item
    \label{en:sparsification-ub-fixed-y}
    If $\vec{a}_1, \ldots, \vec{a}_m \ne \vec{0}$, then for fixed $y \in \F_q$,
    \[
    \Pr_{\vec{x} \sim \F_{q}^{n}}[\exists\ i \in [m] \text{ such that } \iprod{\vec{a}_i, \vec{x}} = y] \leq \frac{m}{q}
    \ \text.
    \]
\end{enumerate}
\end{lemma}

\begin{proof}
We have that $\Pr_{\vec{x} \sim \F_{q}^{n}, y \sim \F_q}[\iprod{\vec{a}_i, \vec{x}} = y] = 1/q$ for each $\vec{a}_i$.
Moreover, if $\vec{a}_i \ne 0$ then $\Pr_{\vec{x} \sim \F_{q}^{n}}[\iprod{\vec{a}_i, \vec{x}} = y] = 1/q$ for any fixed $y$.
The claims then follow by union bound.
\end{proof}

We next show how to sparsify a lattice and argue about vector counts in the resulting sparsified lattice.
In particular, we use \cref{lem:sparsification-lin-eqs,lem:sparsification-upper-bound} to lower bound the minimum distance, lower bound the number of close vectors (and hence also upper bound the close distance), and lower bound the too-close distance in the resulting sparsified lattice.

\begin{proposition}
\label{prop:sparsification-lattice-sets}
Let $p \in [1, \infty)$, let $\lat$ be a lattice of rank $n$ with basis $B$, let $\vec{t} \in \lspan(\lat)$, let $q$ be a prime, and let $r \ge 0$. Let 
$\vec{x}, \vec{z} \sim \F_q^n$ be sampled uniformly at random,
and define
\[
\lat' := \set{\vec{v} \in \lat : \iprod{B^+ \vec{v}, \vec{x}} \equiv 0 \Mod{q}} \ \text, \quad 
\vec{t}' := \vec{t} - B \vec{z} \ \text.
\]
Then the following hold:
\begin{enumerate}
\item
\label{en:sparsification-short}
(Minimum distance.)
If $r \le q \cdot \sm[p]_1(\lat)$, then
$\Pr[\lambda_1(\lat') < r] \leq N_p^o(\lat \setminus \set{\vec{0}}, r, \vec{0})/q$.

\item
\label{en:sparsification-close}
(Close vector count and distance.)
If $r < q \cdot \sm[p]_1(\lat) / 2$, then
for any $\delta > 0$,
\begin{equation*}
    \Pr[N_p(\lat', r, \vec{t}') \le (1 - \delta) N_p(\lat, r, \vec{t}) / q] \le \frac{q}{\delta^2 \cdot N_p(\lat, r, \vec{t})} + \frac{1}{q^n}
    \ \text.
\end{equation*}
In particular, for $\delta = 1$, $\Pr[\dist_p(\vec{t}', \lat') > r] \leq q/N_p(\lat, r, \vec{t}) + 1/q^n$.

\item
\label{en:sparsification-too-close}
(Too-close distance.)
$\Pr[\dist_p(\vec{t}', \lat') < r] \leq N_p^o(\lat, r, \vec{t})/q + 1/q^n$.
\end{enumerate}
\end{proposition}

\begin{proof}
For \cref{en:sparsification-short}, let $m := N_p^o(\lat \setminus \set{\vec{0}}, r, \vec{0})$, let $\vec{v}_1, \ldots, \vec{v}_m \in \lat$ be the $m$ distinct non-zero lattice vectors satisfying $\norm{\vec{v}_i}_p < r$, and let $\vec{a}_i := B^+ \vec{v}_i$ for $i \in [m]$.
Because $r \le q \cdot \sm[p]_1(\lat)$, we know that $\vec{v}_i \notin q \lat$ and thus $\vec{a}_i \not\equiv \vec{0} \Mod{q}$.
Therefore by \cref{lem:sparsification-upper-bound}, \cref{en:sparsification-ub-fixed-y} with $y = 0$,
\begin{align*}
    \Pr[\lambda_1(\lat') < r]
    &= \Pr[\exists \ i \in [m] \text{ such that } \vec{v}_i \in \lat'] \\
    &= \Pr[\exists \ i \in [m] \text{ such that } \iprod{\vec{a}_i, \vec{x}} \equiv 0 \Mod{q}] \\
    &\le m / q
    \ \text.
\end{align*}

For \cref{en:sparsification-close,en:sparsification-too-close} we will use the fact that the statistical distance between $(\vec{x}, \iprod{\vec{z}, \vec{x}})$ and $(\vec{x}, y)$ with $y \sim \F_q$ sampled uniformly at random is $(q - 1) / q^{n+1} < 1/q^n$. This follows by a direct computation and noting that $(\vec{x}, y)$ and $(\vec{x}, \iprod{\vec{z}, \vec{x}})$ are identically distributed conditioned on $\vec{x} \neq \vec{0}$. Additionally, we note that
\begin{align} \label{eq:coset-counts}
N_p(\lat', r, \vec{t}') = N_p(\lat' + B\vec{z}, r, \vec{t})
= N_p(\set{\vec{v} \in \lat : \iprod{B^+ \vec{v}, \vec{x}} \equiv \iprod{\vec{z}, \vec{x}} \Mod{q}}, r, \vec{t}) \ .
\end{align}

For \cref{en:sparsification-close}, let $m := N_p(\lat, r, \vec{t})$, let $\vec{v}_1, \ldots, \vec{v}_m \in \lat$ be the $m$ distinct lattice vectors satisfying $\norm{\vec{t} - \vec{v}_i}_p \le r$, and let $\vec{a}_i := B^+ \vec{v}_i$ for $i \in [m]$.
By triangle inequality, $\norm{\vec{v}_i - \vec{v}_j}_p \le 2 r$ for all $i, j \in [m]$.
Then, because $2 r < q \cdot \sm[p]_1(\lat)$, for $i \neq j$ we know that $\vec{v}_i - \vec{v}_j \notin q \lat$ and thus $\vec{a}_i \not\equiv \vec{a}_j \Mod{q}$.
Let $y \sim \F_q$ be sampled uniformly at random.
Then by \cref{eq:coset-counts} and \cref{lem:sparsification-lin-eqs}, \cref{en:sparsification-distinct},
\begin{align*}
    \Pr[N_p(\lat', r, \vec{t}') \le (1 - \delta) m / q]
    &= \Pr[\card{\set{i \in [m] : \iprod{\vec{a}_i, \vec{x}} \equiv \iprod{\vec{z}, \vec{x}} \Mod{q}}} \le (1 - \delta) m / q] \\
    &\le \Pr[\card{\set{i \in [m] : \iprod{\vec{a}_i, \vec{x}} \equiv y \Mod{q}}} \le (1 - \delta) m / q] + \frac{1}{q^n} \\
    &\le \frac{q}{\delta^2 m} + \frac{1}{q^n}
    \ \text,
\end{align*}

For \cref{en:sparsification-too-close}, let $m := N_p^o(\lat, r, \vec{t})$, let $\vec{v}_1, \ldots, \vec{v}_m \in \lat$ be the $m$ distinct lattice vectors satisfying $\norm{\vec{t} - \vec{v}_i}_p < r$, and let $\vec{a}_i := B^+ \vec{v}_i$ for $i \in [m]$.
Let $y \sim \F_q$ be sampled uniformly at random.
Then by \cref{eq:coset-counts} and \cref{lem:sparsification-upper-bound}, \cref{en:sparsification-ub-random-y},
\begin{equation*}
\begin{aligned}[b]
    \Pr[\dist_p(\vec{t}', \lat') < r]
    &= \Pr[\exists i \in [m] \text{ such that } \iprod{\vec{a}_i - \vec{z}, \vec{x}} \equiv 0 \Mod{q}] \\
    &\le \Pr[\exists i \in [m] \text{ such that } \iprod{\vec{a}_i, \vec{x}} \equiv y \Mod{q}] + \frac{1}{q^n} \\
    &\le \frac{m}{q} + \frac{1}{q^n}
    \ \text.
\end{aligned}
    \qedhere
\end{equation*}
\end{proof}

\subsection{Point Counting via the Theta Function}
\label{subsec:point-counting-Zn}

Here we present results related to counting points in $\Z^n$ using theta functions $\Theta_p$ (which are the $\ell_p$ norm analogs of the Gaussian function). 
The main result presented in this section, \cref{thm:point-counts-thetas}, was originally proved in~\cite{mazo90:_lattice_points} and~\cite{elkies91:_packing_densities}. Here we follow the nomenclature used in~\cite{conf/stoc/AggarwalS18}. 
For $p \in [1, \infty)$, $\tau > 0$, and $t \in \R$, define
\[
\Theta_p(\tau, t) :=
\sum_{z \in \Z} \exp(-\tau \abs{z - t}^p)
\ \text.
\]
We note that without loss of generality we may take $t \in [0, 1/2]$.
For $\vec{t} \in \R^n$, extend this to
\[
\Theta_p(\tau, \vec{t}) := \prod_{i=1}^n \Theta_p(\tau, t_i) = \sum_{z \in \Z^n} \exp(-\tau \norm{\vec{z} - \vec{t}}_p^p)
\ \text.
\]

For $p \in [1, \infty)$, $\tau > 0$, and $t \in \R$, define
\[
\mu_p(\tau, t) 
:= \E_{X\sim D_p(\tau, t)}[\abs{X}^p]
= \frac{1}{\Theta_p(\tau, t)} \cdot \sum_{z \in \Z} \abs{z-t}^p \cdot \exp(-\tau \abs{z-t}^p)
\ \text.
\]
where $D_p(\tau, t)$ is the probability distribution over $\Z - t$ that assigns probability $\exp(-\tau \abs{x}^p)/\Theta_p(\tau, t)$ to $x \in \Z - t$.
For $\vec{t} \in \R^n$, extend this to
\[
\mu_p(\tau, \vec{t}) := \sum_{i=1}^n \mu_p(\tau, t_i)
= \sum_{i=1}^n \E_{X\sim D_p(\tau, t_i)}[\abs{X}^p]
= \sum_{i=1}^n \frac{1}{\Theta_p(\tau, t_i)} \cdot \sum_{z \in \Z}  \abs{z-t_i}^p \cdot \exp(-\tau \abs{z-t_i}^p)
\ \text.
\]

The following fact about $\mu_p$ is claimed but unproven in \cite{conf/stoc/AggarwalS18}.

\begin{claim}
\label{clm:rp-equals-mu}
For any $p \in [1, \infty)$, $\vec{t} \in \R^n$ and $r > \dist_p(\vec{t}, \Z^n)$, there exists a unique $\tau > 0$ such that $r^p = \mu_p(\tau, \vec{t})$.
\end{claim}

\begin{proof}
By calculation (see also \cite[Lemma~6.2]{conf/stoc/AggarwalS18}),
\begin{align*}
    \frac{\partial}{\partial \tau} \Theta_p(\tau, t_i)
    &= \sum_{z \in \Z} -\abs{z - t_i}^p \cdot \exp(-\tau \abs{z - t_i}^p) \\
    &= - \Theta_p(\tau, t_i) \cdot \mu_p(\tau, t_i)
    \ \text, \\
    \frac{\partial}{\partial \tau} \mu_p(\tau, t_i)
    &= \frac{1}{\Theta_p(\tau, t_i)} \sum_{z \in \Z} - \abs{z - t_i}^{2 p} \cdot \exp(-\tau \abs{z - t_i}^p) + \mu_p(\tau, t_i)^2 \\
    &= -\operatorname*{Var}_{X \sim D_p(\tau, t_i)}[\abs{X}^p] < 0
    \ \text.
\end{align*}
Moreover, recalling that $\mu_p(\tau, t_i) = \E_{X \sim D_p(\tau, t_i)}[\abs{X}^p]$, we have $\lim_{\tau \to \infty} \mu_p(\tau, t_i) = \min_{x \in \Z - t_i}[\abs{x}^p] = \dist_p(t_i, \Z)^p$ and $\lim_{\tau \to 0} \mu_p(\tau, t_i) = \infty$ (see \cref{clm:beta-limits}).
Then $\mu_p(\tau, \vec{t}) = \sum_{i = 1}^n \mu_p(\tau, t_i)$ is strictly decreasing in $\tau$ for $\tau \in (0, \infty)$, and $\lim_{\tau \to \infty} \mu_p(\tau, \vec{t}) = \dist_p(\vec{t}, \Z^n)^p$, $\lim_{\tau \to 0} \mu_p(\tau, \vec{t}) = \infty$.
Hence for every $r > \dist_p(\vec{t}, \Z^n)$, there exists a unique $\tau > 0$ such that $\mu_p(\tau, \vec{t}) = r^p$, as desired.
\end{proof}

The following defines functions $\betafun{t}{a}$ that, as we will see, are such that $\betafun{t}{a}^n$ is equal to or a close approximation of $N_p(\Z^n, a n^{1/p}, t \cdot \vec{1})$.

\begin{definition}
\label{def:beta}
For $p \in [1, \infty)$, $t \in [0, 1/2]$, and $a \geq 0$, we define  $\betafun{t}{a}$ as follows.
\begin{enumerate}
    \item For $a < t$, define $\betafun{t}{a} := 0$.
    \item For $a = t$, define $\betafun{1/2}{1/2} := 2$ and for $t \neq 1/2$ define $\betafun{t}{t} := 1$.
    \item For $a > t$, define
    \begin{equation*}
        \betafun{t}{a} := \exp(\tau^* a^p) \cdot \Theta_p(\tau^*, t)
        \ \text,
    \end{equation*}
    where $\tau^* > 0$ is the unique solution to $\mu_p(\tau^*, t) = a^p$.
\end{enumerate}
\end{definition}
We note that $\betafun{t}{a}$ is well-defined in the $a > t$ case by \cref{clm:rp-equals-mu}.
We will also work with the inverse function $\betainv{t}{\nu}$, which we show is well-defined in \cref{clm:beta-properties} for $\nu \geq \betafun{t}{t}$.

The following theorem says that $\betafun{t}{a}^n$ is equal to the number of integer points in a $a n^{1/p}$-scaled, $(t \cdot \vec{1})$-centered $\ell_p$ ball in $n$ dimensions up to a subexponential error term.
This also implies that $a = \betainv{t}{\nu}$ corresponds to the radius $a n^{1/p}$ at which $N_p(\Z^n, an^{1/p}, t \cdot \vec{1}) \approx \nu^n$, where the approximation again holds up to a subexponential error term.
This theorem will be very important for our analysis.
We state the result using the notation of~\cite[Theorem 6.1]{conf/stoc/AggarwalS18}, but again note that it was originally proven for $p = 2$ in~\cite{mazo90:_lattice_points} and for general $p$ in~\cite{elkies91:_packing_densities}.
\begin{theorem} %
\label{thm:point-counts-thetas}
For any constants $p \in [1, \infty)$ and $\tau > 0$, there is another constant $C^* > 0$ such that for any $n \in \Z^+$ and $\vec{t} \in \R^n$,
\[
\exp(\tau \mu_p(\tau, \vec{t}) - C^* \sqrt{n}) \cdot \Theta_p(\tau, \vec{t}) \leq N_p(\Z^n, \mu_p(\tau, \vec{t})^{1/p}, \vec{t}) \leq \exp(\tau \mu_p(\tau, \vec{t})) \cdot \Theta_p(\tau, \vec{t})
\ \text.
\]
\end{theorem}

We conclude this subsection with two technical claims. \cref{clm:beta-properties} gives properties of the function $\betafun{t}{a}$ that we will frequently use in later sections. \cref{clm:beta-limits} argues about several limits involving the functions $\Theta$, $\mu$, and $\beta$, and notes that they do not depend on the parameter $t$.
For readability, the proofs of these claims are deferred to \cref{sec:prelims-claims-proofs}.

\begin{restatable}{claim}{techone}
\label{clm:beta-properties}
For any $p \in [1, \infty)$ and $t \in [0, 1/2]$:
\begin{enumerate}
    \item \label{en:beta-approx-Np} For $a > t$, $\betafun{t}{a}^n \cdot \exp(-C^*\sqrt{n}) \leq N_p(\Z^n, a n^{1/p}, t \cdot \vec{1}) \leq \betafun{t}{a}^n$ for some constant $C^* > 0$;
    \item \label{en:beta-via-min-tau} For $a > t$, $\betafun{t}{a} = \min_{\tau > 0} \exp(\tau a^p) \cdot \Theta_p(\tau, t)$, i.e., the unique solution $\tau^*$ of $\mu_p(\tau^*, t) = a^p$ minimizes $\exp(\tau a^p) \cdot \Theta_p(\tau, t)$;
    \item \label{en:beta-continuity} $\betafun{t}{a}$ is strictly increasing for $a \geq t$ and differentiable (and hence continuous) for $a > t$;
    \item \label{en:beta-inverse-continuity} $\betainv{t}{\nu}$ is well-defined and strictly increasing for $\nu \geq \betafun{t}{t}$, and
    differentiable (and hence continuous) for $\nu > \betafun{t}{t}$.
\end{enumerate}
\end{restatable}

\begin{restatable}{claim}{techtwo}
\label{clm:beta-limits}
For any $p \in [1, \infty)$ and $t \in [0, 1/2]$:
\begin{enumerate}
    \item \label{en:theta-limit}
    $\lim_{\tau \to 0} \Theta_p(\tau, t) \cdot \tau^{1/p} = 2 \Gamma(1 + 1/p)$;
    \item \label{en:mu-limit}
    $\lim_{\tau \to 0} \mu_p(\tau, t) \cdot \tau = 1 / p$;
    \item \label{en:beta-limit}
    $\lim_{a \to \infty} \betafun{t}{a} / a = \lim_{\nu \to \infty} \nu / \betainv{t}{\nu} = 2 \Gamma(1 + 1/p) \cdot (e p)^{1/p}$.
\end{enumerate}
In particular, none of the preceding limits depends on $t$.
\end{restatable}

\subsection{Lattice Problems}

We next introduce the main lattice problems that we will work with, starting with a variant of the Closest Vector Problem ($\GapCVP$).

\begin{definition}
\label{def:gapcvp-prime}
For $p \in [1,\infty]$ and $\gamma \geq 1$, the decision promise problem
$\binGapCVP_{p, \gamma}$ is defined as follows. An instance consists of
a basis $B \in \R^{d \times n}$ and a target vector
$\vec{t} \in \R^d$.
\begin{itemize}
    \item It is a YES instance if there exists $\vec{x} \in \bit^n$ such that $\norm{B\vec{x} - \vec{t}}_p \leq 1$.
    \item It is a NO instance if $\dist_p(\vec{t}, \lat(B)) > \gamma$.
\end{itemize}
\end{definition}

For $\gamma = 1$, we will simply write $\binGapCVP_{p, 1}$ as $\binGapCVP_{p}$. 
We note that for $\binGapCVP$, the distance threshold in the YES case is~$1$ without loss of generality, because we can scale the lattice and target vector. (Our definition of $\GapSVP$ below uses an instance-dependent distance threshold $r > 0$, which will be convenient.) We also note that $\binGapCVP_{p, \gamma}$ is trivially a subproblem of ``standard'' $\GapCVP_{p, \gamma}$ (which we do not define formally here), i.e., an instance of the former is also an instance of the latter.

We next introduce $\BDD$, a search variant of $\CVP$ where the target vector is promised to be close to the lattice.
\begin{definition}
\label{def:bdd}
For $p \in [1, \infty]$ and $\alpha = \alpha(n) > 0$, an instance of the search problem $\BDD_{p,\alpha}$ is a lattice basis $B \in \R^{d \times n}$ and a target $\vec{t} \in \R^d$ satisfying $\dist_p(\vec{t}, \lat(B)) \le \alpha \cdot \sm_1(\lat(B))$, and the goal is to find a closest lattice vector $\vec{v} \in \lat(B)$ to $\vec{t}$ such that $\norm{\vec{t} - \vec{v}}_p = \dist_p(\vec{t}, \lat(B))$.
\end{definition}

We note that there is another frequently used version of $\BDD$ where the goal is instead to find a lattice vector $\vec{v}$ satisfying $\norm{\vec{t} - \vec{v}}_p \le \alpha \cdot \sm_1(\lat(B))$, which is less demanding than the goal in the definition above.
By \cite[Lemma 2.6]{conf/crypto/LyubashevskyM09},\footnote{Formally, as stated, \cite[Lemma 2.6]{conf/crypto/LyubashevskyM09} applies only to the case where $\alpha \geq 1$ and $p = 2$, but a very similar reduction works for all $\alpha > 0$ and $p \geq 1$.} the version of $\BDD$ defined above (in rank $n$) efficiently reduces to this alternative version (in rank $n+1$), so the exponential-time hardness results we prove for the version defined above immediately apply to the alternative version as well.

Finally, we introduce the decision version of the Shortest Vector Problem ($\GapSVP$).

\begin{definition} \label{def:gapsvp}
For $p \in [1, \infty]$ and $\gamma \geq 1$, the decision promise problem $\GapSVP_{p, \gamma}$ is defined as follows. An instance consists of a basis $B \in \R^{d \times n}$ and a distance threshold $r > 0$.
\begin{itemize}
    \item It is a YES instance if $\sm_1(\lat(B)) \leq r$.
    \item It is a NO instance if $\sm_1(\lat(B)) > \gamma r$.
\end{itemize}
\end{definition}
We note that ``$\SVP$'' and ``$\CVP$'' are sometimes instead used to refer to the \emph{search} versions of the Shortest and Closest Vector Problems, with their respective decision problems instead called ``$\problem{GapSVP}$'' and ``$\problem{GapCVP}$.'' There is a trivial reduction from the decision versions to the corresponding search versions of these problems, so any hardness of the former implies identical hardness of the latter. In particular, our main hardness result for decisional $\SVP$ in \cref{thm:gapseth-svp-informal} immediately implies corresponding hardness of search $\SVP$ as well.

\subsection{Hardness Assumptions and Results}
\label{subsec:fine-grained-assumptions}

We will use the following hypotheses, which are ``gap'' variants of the celebrated (Strong) Exponential Time Hypothesis ((S)ETH) of Impagliazzo and Paturi~\cite{journals/jcss/ImpagliazzoP01}. Gap-ETH was introduced in works by Dinur~\cite{journals/eccc/Dinur16} and by Manurangsi and Raghavendra~\cite{conf/icalp/ManurangsiR17}, and Gap-SETH was introduced by Manurangsi in his Ph.D. thesis~\cite{Manurangsi2019}.
Recall that in the $(s, c)$-Gap-$k$-SAT problem with $0 < s \leq c \leq 1$, the goal is to distinguish between $k$-SAT formulas in which at least a $c$ fraction of clauses are satisfiable, and ones in which strictly less than an $s$ fraction of clauses are satisfiable.

\begin{definition}[Gap Exponential Time Hypothesis (Gap-ETH)] \label{def:gap-eth}
There exists $\delta > 0$ such that no $2^{o(n)}$-time algorithm solves $(1 - \delta, 1)$-Gap-$3$-SAT on $n$ variables.
\end{definition}

\begin{definition}[Gap Strong Exponential Time Hypothesis (Gap-SETH)] \label{def:gap-seth}
For every $\eps > 0$ there exist $k \in \Z^+$ and $\delta > 0$ such that no $2^{(1-\eps)n}$-time algorithm solves $(1 - \delta, 1)$-Gap-$k$-SAT on $n$ variables.
\end{definition}

We also define ``plain'' SETH.
\begin{definition}[Strong Exponential Time Hypothesis (SETH)] \label{def:seth}
For every $\eps > 0$ there exists $k \in \Z^+$ such that no $2^{(1-\eps)n}$-time algorithm solves $k$-SAT on $n$ variables.
\end{definition}

We define ``randomized'' (respectively, ``non-uniform'') versions of the above hypotheses analogously, but with ``randomized algorithm'' (respectively, ``non-uniform algorithm'') in place of ``algorithm.'' Here by an $f(n)$-time non-uniform algorithm for solving Gap-$k$-SAT, we mean a family of circuits $\set{C_n}$ such that $C_n$ solves Gap-$k$-SAT on $n$ variables and is of size at most $f(n)$.

We note that, trivially, the randomized versions of these hypotheses are at least as strong as the normal (deterministic) versions, and that, by standard randomness preprocessing arguments such as those used in the proof of Adleman's Theorem (which asserts that $\BPP \subseteq \PPoly$), the non-uniform versions are at least as strong as the randomized versions. See, e.g.,~\cite{conf/coco/BennettP20}.

\subsubsection{Hardness of CVP' Assuming (Gap-)(S)ETH}

We next recall known hardness results about the complexity of $\binGapCVP_{p, \gamma}$ assuming (Gap-)(S)ETH.

\begin{theorem}[{\cite[Corollary 5.2]{conf/focs/BennettGS17}}] \label{thm:gapeth-cvp}
For all $p \in [1, \infty)$ there exists a constant $\gamma = \gamma(p)$ such that $\binGapCVP_{p, \gamma}$ has no $2^{o(n)}$-time algorithm on lattices of rank $n$ assuming Gap-ETH.
\end{theorem}

\begin{theorem}[{\cite[Theorem 4.2]{conf/soda/AggarwalBGS21}}] \label{thm:gapseth-cvp}
For all $p \in [1, \infty) \setminus 2 \Z$ and all $\eps > 0$ there exists a constant $\gamma = \gamma(p, \eps) > 1$ such that there is no $2^{(1-\eps)n}$-time algorithm for $\binGapCVP_{p, \gamma}$ on lattices of rank $n$ assuming Gap-SETH.
\end{theorem}

We also present a variant of the preceding theorem which shows a similar runtime lower bound for $\binGapCVP_{p, \gamma}$ assuming ``plain'' SETH, but for $\gamma$ of the form $\gamma = 1 + 1/\poly(n)$ as opposed to constant $\gamma > 1$.
\begin{theorem}[{\cite[Corollary 3.3]{conf/soda/AggarwalBGS21}}] \label{thm:seth-cvp}
For all $p \in [1, \infty) \setminus 2 \Z$ and all $\eps > 0$ there exists $\gamma = 1 + \Omega(1/n^c)$ for some constant $c > 0$ such that there is no $2^{(1-\eps)n}$-time algorithm for $\binGapCVP_{p, \gamma}$ on lattices of rank $n$ assuming SETH.
\end{theorem}

We remark that all of these theorems are stated for $\GapCVP_{p, \gamma}$ rather than $\binGapCVP_{p, \gamma}$ in the original work, but we note that inspection of their proofs shows that they in fact hold for $\binGapCVP_{p, \gamma}$.
We also remark that the approximation factors~$\gamma$ stated in the original versions of \cref{thm:gapeth-cvp,thm:gapseth-cvp} have an explicit dependence on the soundness and completeness parameters~$s$ and~$c$ in the source Gap-$k$-SAT problems from which the reductions are given. Gap-ETH and Gap-SETH as stated in \cref{def:gap-eth,def:gap-seth} assume soundness $s = 1 - \delta$ for some non-explicit $\delta > 0$ and perfect completeness ($c = 1$), resulting in the non-explicit constant $\gamma > 1$ here.
Finally, we remark that \cref{thm:seth-cvp} is stated in its original version as only showing hardness for exact $\binGapCVP$, but inspection of its proof shows that the result holds with $\gamma = 1 + 1/\poly(n)$, as claimed.

\section{Gap-(S)ETH-hardness of BDD}
\label{sec:gapeth-seth-bdd}

The goal of this section is to prove \cref{thm:gapeth-bdd-kn-informal,thm:gapseth-bdd-informal,thm:gapeth-bdd-int-informal}, which show the Gap-(S)ETH-hardness of $\BDD_{p, \alpha}$.
The common proof idea is to construct reductions from $\binGapCVP_{p, \gamma}$ (whose Gap-(S)ETH-hardness is known; see \cref{thm:gapeth-cvp,thm:gapseth-cvp}) to $\BDD_{p, \alpha}$.

We start by introducing two intermediate problems used in our reductions.
The first problem, $(A, G)$-$\BDD_{p, \alpha}$, is essentially a relaxation of the decision version of $\BDD_{p, a}$ in which there are at least $G$ (``good'') close vectors to the target (at distance at most $r$) in the YES case, at most $A$ (``annoying'') short non-zero vectors (of norm strictly less than $r/\alpha$) in the YES case, and at most $A$ ``too close'' vectors to the target (at distance at most $r$) in the NO case.

\begin{definition}
\label{def:stu-bdd}
Let $A=A(n) \geq 0$, $G=G(n) > A$, $p \in [1, \infty]$, and $\alpha = \alpha(n) > 0$. An instance of the decision promise problem $(A, G)$-$\BDD_{p, \alpha}$ is a lattice basis $B \in \R^{d \times n}$, a target $\vec{t} \in \R^d$, and a distance $r > 0$.
\begin{itemize}
    \item It is a YES instance if $N_p^o(\lat(B) \setminus \set{\vec{0}}, r/\alpha, \vec{0}) \leq A$ and $N_p(\lat(B), r, \vec{t}) \geq G$.
    \item It is a NO instance if $N_p(\lat(B), r, \vec{t}) \leq A$.
\end{itemize}
\end{definition}

We note that $(0, 1)$-$\BDD_{p, \alpha}$ is essentially a decision version of $\BDD_{p, \alpha}$, which is a search problem,
and the decision-to-search reduction from $(0, 1)$-$\BDD_{p, \alpha}$ to $\BDD_{p, \alpha}$ is trivial: call the search oracle and verify if it outputs a lattice vector $\vec{v}$ satisfying $\norm{\vec{t} - \vec{v}}_p \le r$.

\subsection{General Reductions from CVP' to BDD}
\label{sec:reduction-bdd}

Here we show the reductions from $\binGapCVP_{p, \gamma}$ to $(A, G)$-$\BDD_{p, \alpha}$ and $(A, G)$-$\BDD_{p, \alpha}$ to $(0, 1)$-$\BDD_{p, \alpha}$. Since we also have the trivial reduction from $(0, 1)$-$\BDD_{p, \alpha}$ to ``ordinary,'' search $\BDD_{p, \alpha}$, composing these reductions gives us the overall desired reduction from $\binGapCVP_{p, \gamma}$ to $\BDD_{p, \alpha}$.

We start with the following lemma analyzing a transformation involving $\binGapCVP$ instances and a lattice-target gadget $(\latdag, \tdag)$.
This transformation can be regarded as a generalization of the techniques used in \cite{conf/stoc/AggarwalS18, conf/coco/BennettP20}; in particular the lemma generalizes \cite[Lemma~3.5]{conf/coco/BennettP20}.
The transformation is versatile, and the lemma is used in the main reduction in this section as well as the main reduction in \cref{sec:gapseth-svp}.

\begin{lemma}
\label{lem:reduction-gadget}
For any $n' < n$, define the following transformation from a basis $\mat{B}'$ of a rank-$n'$ lattice $\lat'$ to a basis $\mat{B}$ of a rank-$n$ lattice $\lat$, and from a target vector $\vec{t}' \in \lspan(\lat')$ to a target vector $\vec{t} \in \lspan(\lat)$, parameterized by a scaling factor $s > 0$ and a gadget consisting of a rank-$(n - n')$ lattice $\lat^\dagger = \lat(\mat{B}^\dagger)$ and a target vector $\vec{t}^\dagger$:
\[
    \mat{B} := \begin{pmatrix} s \mat{B}' & \mat{0} \\ \mat{I}_{n'} & \mat{0} \\ \mat{0} & \mat{B}^\dagger \end{pmatrix}
    \ \text, \qquad
    \vec{t} := \begin{pmatrix} s \vec{t}' \\ \tfrac{1}{2} \vec{1}_{n'} \\ \vec{t}^\dagger \end{pmatrix}
    \ \text.
\]
Then for any $p \in [1, \infty)$, $r > 0$, and $\gamma \ge 1$,
\begin{enumerate}
    \item
    \label{en:gadget-short-ub}
    $N_p^o(\lat, r, (\vec{t}_1, \vec{t}_2)) \le N_p^o(\Z^{n'} \oplus \lat^\dagger, r, \vec{t}_2)$ for any $\vec{t}_1 \in \lspan(\lat'), \vec{t}_2 \in \lspan(\Z^{n'} \oplus \lat^\dagger)$;
    \item
    \label{en:gadget-close-lb}
    if there exists $\vec{x} \in \set{0,1}^{n'}$ such that $\norm{\mat{B}' \vec{x} - \vec{t}'}_p \le 1$, then $N_p(\lat, (s^p + n' / 2^p + r^p)^{1/p}, \vec{t}) \ge N_p(\lat^\dagger, r, \vec{t}^\dagger)$;
    \item
    \label{en:gadget-too-close-ub}
    if $\dist_p(\vec{t}', \lat') > \gamma$, then $N_p(\lat, (\gamma^p s^p + r^p)^{1/p}, \vec{t}) \le N_p^o(\Z^{n'} \oplus \lat^\dagger, r, (\tfrac{1}{2} \vec{1}_{n'}, \vec{t}^\dagger))$.
\end{enumerate}
\end{lemma}

\begin{proof}
\cref{en:gadget-short-ub} follows from the observations that every vector $\vec{v} = (s \mat{B}' \vec{x}, \vec{x}, \mat{B}^\dagger \vec{y}) \in \lat$ for $\vec{x} \in \Z^{n'}, \vec{y} \in \Z^{n - n'}$ corresponds bijectively to the vector $\vec{v}' = (\vec{x}, \mat{B}^\dagger \vec{y}) \in \Z^{n'} \oplus \lat^\dagger$ and that $\norm{\vec{v}' - \vec{t}_2}_p \leq \norm{\vec{v} - (\vec{t}_1, \vec{t}_2)}_p$.

For \cref{en:gadget-close-lb}, note that for every $\vec{y} \in \Z^{n - n'}$ such that $\norm{\mat{B}^\dagger \vec{y} - \vec{t}^\dagger}_p \le r$, the lattice vector $\vec{v} = (s \mat{B}' \vec{x}, \vec{x}, \mat{B}^\dagger \vec{y}) \in \lat$ satisfies
\[
    \norm{\vec{v} - \vec{t}}_p^p = s^p \norm{\mat{B}' \vec{x} - \vec{t}'}_p^p + \norm{\vec{x} - \tfrac{1}{2} \vec{1}_{n'}}_p^p + \norm{\mat{B}^\dagger \vec{y} - \vec{t}^\dagger}_p^p
    \le s^p + n' / 2^p + r^p
    \ \text.
\]
Hence the claim follows.

For \cref{en:gadget-too-close-ub}, similarly, for every $\vec{x} \in \Z^{n'}, \vec{y} \in \Z^{n - n'}$ such that the lattice vector $\vec{v} = (s \mat{B}' \vec{x}, \vec{x}, \vec{y}) \in \lat$ satisfies $\norm{\vec{v} - \vec{t}}_p \le (\gamma^p s^p + r^p)^{1/p}$, we have
\[
    \norm{\vec{x} - \tfrac{1}{2} \vec{1}_{n'}}_p^p + \norm{\mat{B}^\dagger \vec{y} - \vec{t}^\dagger}_p^p = \norm{\vec{v} - \vec{t}}_p^p - s^p \norm{\mat{B}' \vec{x} - \vec{t}'}_p^p < (\gamma^p s^p + r^p) - s^p \cdot \gamma^p = r^p
    \ \text.
\]
Hence the claim follows.
\end{proof}

We then show the reduction from $(A, G)$-$\BDD$ to $(0, 1)$-$\BDD$.
A similar reduction exists in \cite{conf/coco/BennettP20}. In particular, \cref{lem:reduction-decision-bdd} uses sparsification in a similar way as \cite[Theorem~3.3]{conf/coco/BennettP20} while targeting the $(A, G)$-$\BDD$ problem.
We note in passing that the $\alpha \ge 1/2$ constraint in \cref{lem:reduction-decision-bdd} is not essential but simplifies some expressions and suffices for our purposes.

\begin{lemma}
\label{lem:reduction-decision-bdd}
For any $p \in [1, \infty)$, $\alpha \ge 1/2$, 
and $A = A(n) \ge 1$, $G = G(n) \ge \ratioGA A$, where $A(n)$ is computable in $\poly(n)$ time, there is a rank-preserving randomized Karp reduction with bounded error from $(A, G)$-$\BDD_{p, \alpha}$ to $(0, 1)$-$\BDD_{p, \alpha}$.
\end{lemma}

\begin{proof}
The reduction works as follows.
On input an $(A, G)$-$\BDD_{p, \alpha}$ instance $(\mat{B}, r, \vec{t})$,
it first samples a prime $20 \alpha A \le q \le 40 \alpha A$, 
It next samples $\vec{x}, \vec{z} \in \F_q^n$ uniformly at random and sets
\[
    \lat' := \{\vec{v} \in \lat : \iprod{\mat{B}^+ \vec{v}, \vec{x}} \equiv 0 \Mod{q}\}
    \ \text, \qquad
    \vec{t}' := \vec{t} - \mat{B} \vec{z}
    \ \text.
\]
Then, it computes a basis $\mat{B}'$ be a basis of $\lat'$. Such a $\mat{B}'$ is efficiently computable given $\mat{B}, q, \vec{x}$; see, e.g., \cite[Claim~2.15]{conf/soda/Stephens-Davidowitz16}). Finally, it outputs $(\mat{B}', r, \vec{t}')$. Clearly, the reduction is efficient so it remains to show correctness.

First note that by \cref{clm:count-properties}, \cref{en:count-short-sm}, we know $A \ge 2 r / (\alpha \cdot \sm_1(\lat(\mat{B}))) - 2$ and thus $r / \alpha \le 2 r \le \alpha (A + 2) \cdot \sm_1(\lat(\mat{B})) < q \cdot \sm_1(\lat(\mat{B}))$.
If $(\mat{B}, r, \vec{t})$ is a YES instance, i.e., $N_p^o(\lat(\mat{B}) \setminus \set{\vec{0}}, r/\alpha, \vec{0}) \le A$ and $N_p(\lat(\mat{B}), r, \vec{t}) \ge G$, then by \cref{prop:sparsification-lattice-sets}, \cref{en:sparsification-short} (using $r / \alpha$ for $r$), $N_p^o(\lat(\mat{B}') \setminus \set{\vec{0}}, r/\alpha, \vec{0}) > 0$ with probability at most $A / q$.
Moreover, by \cref{prop:sparsification-lattice-sets}, \cref{en:sparsification-close}, $N_p(\lat(\mat{B}'), r, \vec{t}') = 0$ with probability at most $q / G + 1/q^n$.
Hence by union bound, $(\mat{B}', r, \vec{t}')$ is a YES instance for $(0, 1)$-$\BDD_{p, \alpha}$ with probability at least
\begin{equation*}
    1 - \frac{A}{q} - \frac{q}{G} - \frac{1}{q^n}
    \ge 1 - \frac{A}{20 \alpha A} - \frac{40 \alpha A}{\ratioGA A} - \frac{1}{20 \alpha A}
    \ge 1 - \frac{1}{10} - \frac{1}{10} - \frac{1}{10}
    > \frac{2}{3}
    \ \text.
\end{equation*}

If $(\mat{B}, r, \vec{t})$ is a NO instance of $(A, G)$-$\BDD_{p, \alpha}$ (i.e., $N_p(\lat(\mat{B}), r, \vec{t}) \le A$) then by \cref{prop:sparsification-lattice-sets}, \cref{en:sparsification-too-close}, $(\mat{B}', r, \vec{t}')$ is a NO instance of $(0, 1)$-$\BDD_{p, \alpha}$ (i.e., $N_p(\lat(\mat{B}'), r, \vec{t}) = 0$) with probability at least
\[
    1 - \frac{A}{q} - \frac{1}{q^n}
    \ge 1 - \frac{1}{10} - \frac{1}{10}
    > \frac{2}{3}
    \ \text.
\]
In each case, the reduction succeeds with probability at least $2/3$, as needed.
\end{proof}

We conclude this subsection with a theorem that analyzes the general class of reductions from $\binGapCVP$ to $\BDD$ obtained by instantiating \cref{lem:reduction-gadget} with some lattice-target gadget $(\lat^{\dagger}, \vec{t}^{\dagger})$ and then applying \cref{lem:reduction-decision-bdd}.
In particular, the theorem requires there to be roughly $\nu_0^n$ times as many close lattice vectors $N_p(\lat^\dagger, \alpha_G, \vec{t}^\dagger)$ in the gadget as short lattice vectors $N_p(\lat^\dagger, 1, \vec{0})$, and roughly $\nu_1^n$ times as many close lattice vectors as ``too close'' lattice vectors $N_p(\lat^\dagger, \alpha_A, \vec{t}^\dagger)$. 

More specifically, the theorem gives a reduction to $\BDD_{p, \alpha}$ with explicit multiplicative rank increase $C > 1$ for any $\alpha$ larger than some expression involving the input $\binGapCVP_{p, \gamma}$ instance approximation factor $\gamma$, the gadget ``close  distance'' $\alpha_G$, the gadget ``too close distance'' $\alpha_A$, the close/short vector count gap factor $\nu_0$, and the close/``too close'' vector count gap factor $\nu_1$; see \cref{eq:generic-alpha-seth}. (The gadget is normalized so that its ``short distance'' is equal to $1$.) Additionally, the theorem gives a reduction to $\BDD_{p, \alpha}$ for any $\alpha$ larger than a simpler expression not depending on $\nu_0$ or $\nu_1$ with linear but non-explicit rank increase; see \cref{eq:generic-alpha-eth}. \cref{eq:generic-alpha-seth} leads to SETH-type hardness results, while \cref{eq:generic-alpha-eth} leads to ETH-type hardness results.
All three of our main results for $\BDD$, \cref{thm:gapeth-bdd-kn-informal,thm:gapseth-bdd-informal,thm:gapeth-bdd-int-informal}, work by instantiating this theorem.

\begin{theorem}
\label{thm:generic-bdd}
Suppose that for $p \in [1, \infty)$ there exist constants $\alpha_G \ge \alpha_A > 0$, $\nu_0, \nu_1 > 1$ and a family of bases and targets $\set{(\mat{B}^\dagger_n, \vec{t}^\dagger_n)}_{n = 1}^\infty$ satisfying $\mat{B}^\dagger_n \in \R^{m \times n}$ for $m = \poly(n)$, $\vec{t}^\dagger_n \in \lspan(\lat^\dagger_n)$ where $\lat^\dagger_n = \lat(\mat{B}^\dagger_n)$, and
\begin{equation}
    N_p(\lat^\dagger_n, \alpha_G + o(1), \vec{t}^\dagger_n) \ge \max \bigl\{ \nu_0^{n - o(n)} \cdot N_p^o(\lat^\dagger_n, 1, \vec{0}), \nu_1^{n - o(n)} \cdot N_p^o(\lat^\dagger_n, \alpha_A, \vec{t}^\dagger_n) \bigr\}
    \ \text.
    \label{eq:gadget-assumption}
\end{equation}
Then for any constants $C > 1$, $\gamma > 1$, and
\begin{equation} \label{eq:generic-alpha-seth}
    \alpha > \Bigl( \alpha_G^p + \frac{\alpha_G^p - \min\{\alpha_A^p, ((2 d_1)^p - 1) / (2 d_0 )^p\}}{\gamma^p - 1} + \frac{1}{(2 d_0)^p} \Bigr)^{1/p}
    \ \text,
\end{equation}
where $d_0 = \betainv{0}{\nu_0^{C-1}}$ and $d_1 = \betainv{1/2}{\max\set{\nu_1^{C-1}, 2}}$,\footnote{We note that $\betainv{1/2}{\nu}$ is only defined for $\nu \geq \betafun{1/2}{1/2} = 2$, and that the second term in this max corresponds to $\betainv{1/2}{2} = 1/2$.}
there exists an efficient non-uniform reduction from $\binGapCVP_{p, \gamma}$ in rank $n'$ to $\BDD_{p, \alpha}$ in rank $C n'$ for all sufficiently large $n'$.

Moreover, for any constants $\gamma > 1$ and
\begin{equation} \label{eq:generic-alpha-eth}
    \alpha > \Bigl( \alpha_G^p + \frac{\alpha_G^p - \min\{\alpha_A^p, 1\}}{\gamma^p - 1} \Bigr)^{1/p}
    \ \text,
\end{equation}
there exist a constant $C > 1$ and an efficient non-uniform reduction from $\binGapCVP_{p, \gamma}$ in rank $n'$ to $\BDD_{p, \alpha}$ in rank $C n'$ for all sufficiently large $n'$.

Furthermore, if $(B_n^{\dagger}, \vec{t}_n^{\dagger})$ can be computed in $\poly(n)$ randomized time and each term in \cref{eq:gadget-assumption} can be approximated to within a $2^{o(n)}$ factor in $\poly(n)$ randomized time, then the preceding efficient non-uniform reductions can be replaced by efficient randomized reductions.
\end{theorem}

\begin{proof}
By \cref{lem:reduction-decision-bdd} as well as the trivial reduction from $(0, 1)$-$\BDD_{p, \alpha}$ to (ordinary, search) $\BDD_{p, \alpha}$, it suffices to give a reduction that, on input a $\binGapCVP_{p, \gamma}$ instance $(\mat{B}', \vec{t}')$ of rank $n'$, outputs an $(A, G)$-$\BDD_{p, \alpha}$ instance $(\mat{B}, \vec{t}, r)$ of rank $n := Cn'$ with $G \ge \ratioGA A$ for all sufficiently large $n'$.
To do this, 
we apply the transformation in \cref{lem:reduction-gadget} with scale factor $s$ and basis-target gadget $(\ell \mat{B}^\dagger_{n - n'}, \ell \vec{t}^\dagger_{n - n'})$. I.e, on input $(B', \vec{t}')$ the reduction sets
\begin{equation*}
    \mat{B} := \begin{pmatrix} s \mat{B}' & \mat{0} \\ \mat{I}_{n'} & \mat{0} \\ \mat{0} & \ell \mat{B}^\dagger_{n - n'} \end{pmatrix}
    \ \text, \qquad
    \vec{t} := \begin{pmatrix} s \vec{t}' \\ \tfrac{1}{2} \vec{1}_{n'} \\ \ell \vec{t}^\dagger_{n - n'} \end{pmatrix}
    \ \text,
\end{equation*}
and certain $r$. We will give the choice of the parameters $r, s, \ell$ later.
Note that the reduction is (non-uniformly) efficient given that the gadget $(\ell \mat{B}^\dagger_{n - n'}, \ell \vec{t}^\dagger_{n - n'})$ is provided as advice, $\lat^\dagger_{n - n'}$ has rank $n - n'$ and dimension $m$ that are both polynomial in $n'$, and, as will be clear from their definitions later, the parameters $r, s, \ell$ are efficiently computable.

By \cref{lem:reduction-gadget}
(using $r / \alpha$, $(r^p - s^p - n' / 2^p)^{1/p}$, and $(r^p - \gamma^p s^p)^{1/p}$ for $r$ in \cref{en:gadget-short-ub,en:gadget-close-lb,en:gadget-too-close-ub} in the lemma, respectively),
this reduction from $\binGapCVP_{p, \gamma}$ to $(A, G)$-$\BDD_{p, \alpha}$ is correct for
\begin{align*}
    A &:= \begin{aligned}[t]
        \max \bigl\{
            & N_p^o(\Z^{n'} \oplus \ell \lat^\dagger_{n - n'}, r / \alpha, \vec{0}), \\
            & N_p^o(\Z^{n'} \oplus \ell \lat^\dagger_{n - n'}, (r^p - \gamma^p s^p)^{1/p}, (\tfrac{1}{2} \vec{1}, \ell \vec{t}^\dagger_{n - n'}))
        \bigr\}
        \ \text,
    \end{aligned} \\
    G &:= N_p(\ell \lat^\dagger_{n - n'}, (r^p - s^p - n' / 2^p)^{1/p}, \ell \vec{t}^\dagger_{n - n'})
    \ \text.
\end{align*}
By \cref{clm:count-direct-sum}, we have
\begin{equation}
    \begin{aligned}
        A \le \max \bigl\{
            & N_p^o(\Z^{n'}, r / \alpha, \vec{0}) \cdot N_p^o(\ell \lat^\dagger_{n - n'}, r / \alpha, \vec{0}), \\
            & N_p^o(\Z^{n'}, (r^p - \gamma^p s^p)^{1/p}, \tfrac{1}{2} \vec{1}) \cdot N_p^o(\ell \lat^\dagger_{n - n'}, (r^p - \gamma^p s^p - n' / 2^p)^{1/p}, \ell \vec{t}^\dagger_{n - n'})
        \bigr\}
        \ \text.
    \end{aligned}
    \label{eq:generic-bdd-A-ub}
\end{equation}
Then in order for $G \ge \ratioGA A$ to hold, it suffices to set parameters $r, s, \ell$ so that the following inequalities hold:
\begin{align}
    r / \alpha &\le \ell
    \ \text, \label{eq:generic-bdd-short-radii} \\
    (r^p - \gamma^p s^p - n' / 2^p)^{1/p} &\le \alpha_A \cdot \ell
    \ \text, \label{eq:generic-bdd-close-radii} \\
    (r^p - s^p - n' / 2^p)^{1/p} &\ge (\alpha_G + o(1)) \cdot \ell
    \ \text, \label{eq:generic-bdd-too-close-radii} \\
    \nu_0^{(C - 1) n' - o(n')} &\ge \ratioGA \cdot N_p^o(\Z^{n'}, r / \alpha, \vec{0})
    \ \text, \label{eq:generic-bdd-short} \\
    \nu_1^{(C - 1) n' - o(n')} &\ge \ratioGA \cdot N_p^o(\Z^{n'}, (r^p - \gamma^p s^p)^{1/p}, \tfrac{1}{2} \vec{1})
    \ \text. \label{eq:generic-bdd-too-close}
\end{align}
The three inequalities in \cref{eq:generic-bdd-short-radii,eq:generic-bdd-close-radii,eq:generic-bdd-too-close-radii} ensure that we can apply the condition given in \cref{eq:gadget-assumption} to the $\ell$-scaled gadget $(\ell \mat{B}^\dagger_{n - n'}, \ell \vec{t}^\dagger_{n - n'})$. In particular, the expressions on the right-hand sides of these inequalities correspond to $\ell$ times the gadget ``short distance'' $1$, the gadget ``too close distance'' $\alpha_A$, and the gadget ``close distance'' $\alpha_G + o(1)$, respectively.
The two inequalities in \cref{eq:generic-bdd-short,eq:generic-bdd-too-close} correspond to the two terms in the upper bound on $A$ in \cref{eq:generic-bdd-A-ub}, and by \cref{eq:gadget-assumption,eq:generic-bdd-short-radii,eq:generic-bdd-close-radii,eq:generic-bdd-too-close-radii} suffice to ensure that $G \ge \ratioGA A$.

Now let $r := a {n'}^{1/p}$ and $s := b {n'}^{1/p}$ for some constants $a, b > 0$ to be determined later, and set $\ell := d_0 {n'}^{1/p}$.
First, for \cref{eq:generic-bdd-short} to hold, by \cref{clm:beta-properties}, \cref{en:beta-approx-Np}, it suffices to have
\begin{equation*}
    \nu_0^{(C - 1) n' - o(n')} \ge \ratioGA \cdot \betafun{0}{a / \alpha}^{n'}
    \ \text,
\end{equation*}
which in turn holds for all sufficiently large $n'$ if
\begin{equation*}
    \nu_0^{C-1} > \betafun{0}{a / \alpha}
    \ \text.
\end{equation*}
Equivalently, by the (strict) monotonicity of $\betafun{0}{\cdot}$ it suffices to have
\begin{equation*}
    a / \alpha < \betainv{0}{\nu_0^{C-1}} = d_0
    \ \text.
\end{equation*}
Similarly, for \cref{eq:generic-bdd-too-close} to hold for all sufficiently large $n'$, it suffices to have $(a^p - \gamma^p b^p)^{1/p} < d_1$.\footnote{Note that when $\nu_1^{C-1} < 2$ we have $d_1 = \betainv{1/2}{2} = 1/2$ by definition, in which case $(a^p - \gamma^p b^p)^{1/p} < d_1$ ensures that $\betafun{1/2}{(a^p - \gamma^p b^p)^{1/p}} = 0 < \nu_1^{C-1}$.}
Moreover, \cref{eq:generic-bdd-too-close-radii} holds for all sufficiently large $n'$ as long as $(a^p - b^p - 1 / 2^p)^{1/p} > \alpha_G d_0$.

Putting it all together, \cref{eq:generic-bdd-short-radii,eq:generic-bdd-close-radii,eq:generic-bdd-too-close-radii,eq:generic-bdd-short,eq:generic-bdd-too-close} hold for all sufficiently large $n'$ if the following four constraints hold:
\begin{equation}
    \begin{aligned}
        a / \alpha &< d_0
        \ \text, \\
        (a^p - \gamma^p b^p - 1 / 2^p)^{1/p} &< \alpha_A d_0
        \ \text, \\
        (a^p - b^p - 1 / 2^p)^{1/p} &> \alpha_G d_0
        \ \text, \\
        (a^p - \gamma^p b^p)^{1/p} &< d_1
        \ \text.
    \end{aligned}
    \label{eq:main-reduction-bdd-four-ineqs}
\end{equation}
More specifically, the first constraint in \cref{eq:main-reduction-bdd-four-ineqs} implies both \cref{eq:generic-bdd-short-radii,eq:generic-bdd-short} (since we set $\ell = d_0 {n'}^{1/p}$), the second constraint implies \cref{eq:generic-bdd-close-radii}, the third constraint implies \cref{eq:generic-bdd-too-close-radii}, and the last constraint implies \cref{eq:generic-bdd-too-close}.
The last three constraints in \cref{eq:main-reduction-bdd-four-ineqs} are equivalent to the single constraint that
\begin{equation}
\frac{a^p - \min \set{1/2^p + (\alpha_A d_0)^p, d_1^p}}{\gamma^p} < b^p < a^p - 1/2^p - \alpha_G^p d_0^p \ ,
\label{eq:main-reduction-ab}
\end{equation}
and one can check that for every 
\begin{equation}
a > \Bigl( (\alpha_G d_0)^p + \frac{(\alpha_G d_0)^p - \min\{(\alpha_A d_0)^p, d_1^p - 1/2^p\}}{\gamma^p - 1} + \frac{1}{2^p} \Bigr)^{1/p}
\ \text,
\label{eq:main-reduction-a-lb}
\end{equation}
the left-hand side of \cref{eq:main-reduction-ab} is strictly less than the right-hand side of \cref{eq:main-reduction-ab}. I.e., for every $a$ satisfying \cref{eq:main-reduction-a-lb}
there exists $b > 0$ such that $a$ and $b$ satisfy \cref{eq:main-reduction-ab}.
Dividing both sides of \cref{eq:main-reduction-a-lb} by $d_0$ and combining it with the first constraint in \cref{eq:main-reduction-bdd-four-ineqs} (i.e., $\alpha > a / d_0$), we get that for all $\alpha$ satisfying \cref{eq:generic-alpha-seth} there exist $a$ and $b$ such that $G \geq \ratioGA A$ holds for all sufficiently large $n'$, as needed.

For the additional result associated with \cref{eq:generic-alpha-eth}, assume without loss of generality that $\nu_0 = \nu_1$, as otherwise one could use $\min\set{\nu_0, \nu_1}$ in place of both $\nu_0$ and $\nu_1$ in \cref{eq:gadget-assumption}. Then the result follows by taking the limit $C \to \infty$ in \cref{eq:generic-alpha-seth}, and noting that by \cref{clm:beta-limits} we have $d_0, d_1 \to \infty$ and $d_0 / d_1 \to 1$ when $C \to \infty$.
\end{proof}

\subsection{Gap-ETH-hardness of BDD from Exponential Kissing Number Lattices}
\label{sec:gapeth-bdd-kn}

\begin{definition}[Lattice kissing number] \label{def:lattice-kissing-number}
The \emph{kissing number} $\tau(\lat)$ of a lattice $\lat$ is defined as
\begin{equation*}
    \tau(\lat) := N_2(\lat \setminus \set{\vec{0}}, \lambda_1(\lat), \vec{0})
    \ \text.
\end{equation*}
The \emph{maximum lattice kissing number} $\tau^l_n$ for rank $n$ is defined as
\begin{equation*}
    \tau^l_n := \max_{\rank(\lat) = n} \tau(\lat)
    \ \text.
\end{equation*}
\end{definition}

It is a longstanding open question whether there exists an infinite family of lattices with exponential kissing number, i.e., whether there exist infinitely many $n$ such that $\log_2(\tau^l_n) / n$ is greater than a constant. Accordingly, we define \emph{the $\ell_2$ kissing-number constant} to be
\begin{equation} \label{eq:kiss-num-def}
\kissnum := \liminf_{n \to \infty} \log_2(\tau^l_n) / n \ \text{.}
\end{equation}
A main claim in~\cite{Vladut19} was that $\kissnum$ was at least an explicit positive constant, i.e., that there exists a family $\set{\lat_n}_{n=1}^{\infty}$ of \emph{exponential kissing number lattices}. However, that work turned out to have a bug~\cite{bennett2024difficulties}. So, it it remains an open question whether $\kissnum > 0$, and the main gadget constructions in this section, which are based on the existence of exponential kissing number lattices, are conditional.

Specifically, in this section, we show gadgets for instantiating the general reductions in \cref{thm:generic-bdd} based on exponential kissing number lattices (see \cref{cor:gadget-eth-bdd-kn} for the resulting gadgets, which are obtained by sequentially applying \cref{lem:decrease-close-radius,lem:increase-too-close-radius,lem:norm-embedding} to the exponential kissing number lattices).

The following lemma shows that, given a rank-$n$ lattice $\lat^\dagger$ and a target $\vec{t}^\dagger$, by moving $\vec{t}^\dagger$ by a small distance $\delta$ in a random direction to some new target $\vec{t}'$ and picking a distance that is $1 - \eps$ times smaller than the original distance, the number of vectors within (relative) distance $1 - \eps$ of $\vec{t}'$ is in expectation roughly a $\sin(\theta)^n$ factor less than the number of vectors within the original distance of $\vec{t}$, where $\theta$ is (only) a function of $\delta$ and $\eps$.
This allows us to construct gadgets with smaller relative distances $\alpha_G$ 
(where $\alpha_G$ is as in \cref{thm:generic-bdd})
at the cost of an exponential decrease $\sin(\theta)^n$ in the number of close vectors $N_p(\latdag, \alpha_G, \tdag)$; fortunately we can afford this decrease as long as the original count is exponential.

To remark, the lemma is similar to~\cite[Corollary~5.8]{conf/stoc/AggarwalS18}, but has a larger parameter space for $\delta$ and $\eps$, and shows that a larger fraction (essentially $\sin(\theta)^n$) of close vectors is preserved.

\begin{lemma}
\label{lem:decrease-close-radius}
For any lattice $\lat^\dagger$ of rank $n$, target $\vec{t}^\dagger \in \lspan(\lat^\dagger)$, and constants $\eps \in (0, 1/2]$ and $\delta \in [\eps, 1 - \eps]$,
there exists $\vec{t}' \in \lspan(\lat^\dagger)$ with $\norm{\vec{t}' - \vec{t}^\dagger}_2 = \delta$,
such that
\begin{equation*}
    N_2(\lat^\dagger, 1 - \eps, \vec{t}') \ge \frac{\sin(\theta)^n}{\poly(n)} \cdot N_2(\lat^\dagger, 1, \vec{t}^\dagger)
    \ \text,
\end{equation*}
where $\theta := \arccos(\frac{\delta^2 + 2 \eps - \eps^2}{2 \delta})$.
\end{lemma}

\begin{proof}
Without loss of generality suppose that $\lat^\dagger$ has dimension $n$ so that $\lspan(\lat^\dagger) = \R^n$.
It suffices to show that the expectation of $N_2(\lat^\dagger, 1 - \eps, \vec{t}')$ for a uniformly random target $\vec{t}'$ sampled from the sphere $\delta \cdot S^{n-1} + \vec{t}^\dagger$ is at least $(\sin(\theta)^n / \poly(n)) \cdot N_2(\lat^\dagger, 1, \vec{t}^\dagger)$, and therefore by linearity of expectation it suffices to show that the probability that $\norm{\vec{t}' - \vec{v}}_2 \le 1 - \eps$ is at least $\sin(\theta)^n / \poly(n)$ for any vector $\vec{v}$ with $\norm{\vec{v} - \vec{t}^\dagger} \le 1$.
Let $\vec{a} = \vec{t}' - \vec{t}^\dagger$, $\vec{b} = \vec{v} - \vec{t}^\dagger$, and $\theta'$ be the angle between $\vec{a}$ and $\vec{b}$.
We know $\norm{\vec{a}}_2 = \delta$ and $\norm{\vec{b}}_2 \le 1$.
We calculate the probability as follows.
\begin{equation*}
    \begin{aligned}
        \Pr_{\vec{t}'}[\norm{\vec{t}' - \vec{v}}_2 \le 1 - \eps]
        &= \Pr_{\vec{t}'}[\norm{\vec{a}}_2^2 + \norm{\vec{b}}_2^2 - 2 \norm{\vec{a}}_2 \norm{\vec{b}}_2 \cos(\theta') \le (1 - \eps)^2] \\
        &= \Pr_{\vec{t}'} \Bigl[ \cos(\theta') \ge \frac{\norm{\vec{b}}_2^2 - ((1 - \eps)^2 - \delta^2)}{2 \delta \norm{\vec{b}}_2} \Bigr] \\
        &\ge \Pr_{\vec{t}'} \Bigl[ \cos(\theta') \ge \frac{1 - ((1 - \eps)^2 - \delta^2)}{2 \delta} \Bigr] \\
        &= \Pr_{\vec{t}'}[\theta' \le \theta] \\
        &= C_{n-1}(1 - \cos(\theta)) / A_{n-1}
        \ \text,
    \end{aligned}
\end{equation*}
where $C_{n-1}(h)$ is the surface area of a spherical cap with height $h$ of the unit sphere $S^{n-1}$, and $A_{n-1}$ is the surface area of the entire unit sphere $S^{n-1}$.
The ratio in the last quantity satisfies $C_{n-1}(1 - \cos(\theta)) / A_{n-1} = \sin(\theta)^n/\Theta(n^{3/2})$; see, e.g., \cite[Lemma~2.1 and Appendix~A]{conf/soda/BeckerDGL16}.
\end{proof}

The following lemma is implicit in \cite[Lemma~5.2]{conf/stoc/AggarwalS18}, with small changes regarding $N_p$ and $N_p^o$. It uses an averaging argument to show that if there exists a gadget $(\lat^{\dagger}, \vec{t}^{\dagger})$ with at least $M$ times as many close vectors $N_p(\lat^{\dagger}, 1 - \eps, \vec{t}^\dagger)$ as short vectors $N_p(\lat^{\dagger}, 1, \vec{0})$, then there exist a (potentially different) target $\vec{t}'$ and $\eps' \geq \eps$ such that there are at least $M' \approx M^{1/\log_{\eta}(2\eps)}$ times as many close vectors $N_p(\lat^{\dagger}, 1 - \eps', \vec{t}')$ as short vectors $N(\lat^{\dagger}, 1, \vec{0})$ \emph{and in addition} at least $M'$ as many close vectors as ``too close'' vectors $N_p(\lat^{\dagger}, 1 - \eps'/\eta, \vec{t}')$. Here $\eta \in [2 \eps, 1)$ is a parameter that controls the trade-off between $M'$ and the difference $\eps'/\eta - \eps'$ of the distances for the close and the ``too close'' lattice vectors.
In particular, taking larger $\eta$ allows us to decrease the difference between the ``close distance'' $\alpha_G$ and ``too close distance'' $\alpha_A$ in \cref{thm:generic-bdd}
at the cost of having a smaller factor $M'$ in place of $M$ (this trade-off is analogous to the one in \cref{lem:decrease-close-radius}, which allowed us to decrease $\alpha_G$).

\begin{lemma}
\label{lem:increase-too-close-radius}
For any $p \in [1, \infty)$, lattice $\lat^\dagger$, target $\vec{t}^\dagger \in \lspan(\lat^\dagger)$, and constants $\eps \in (0, 1/2)$ and $\eta \in [2 \eps, 1)$ satisfying $N_p(\lat^\dagger, 1 - \eps, \vec{t}^\dagger) \ge M \cdot N_p^o(\lat^\dagger, 1, \vec{0})$ for some $M > 1$,
there exist $\eps' \in [\eps, 1/2]$ and $\vec{t}' \in \lspan(\lat^\dagger)$,
such that
\begin{equation*}
    N_p(\lat^\dagger, 1 - \eps', \vec{t}') \ge M' \cdot \max \bigl\{ N_p^o(\lat^\dagger, 1, \vec{0}), N_p(\lat^\dagger, 1 - \eps' / \eta, \vec{t}') \bigr\}
    \ \text,
\end{equation*}
where $M' := M^{1/(k+1)}$ and $k := \lfloor \log_\eta(2 \eps) \rfloor$.
\end{lemma}

\begin{proof}
For $i = 0, 1, \dots, k + 1$, define $\eps_i = \eps / \eta^i$ and $D_i = \max_{\vec{t} \in \lspan(\lat^\dagger)} N_p(\lat^\dagger, 1 - \eps_i, \vec{t})$.
Write $N = N_p^o(\lat^\dagger, 1, \vec{0})$ for simplicity.
By the given assumption, we know $D_0 \ge N_p(\lat^\dagger, 1 - \eps, \vec{t}^\dagger) \ge M \cdot N$.
Moreover, $k = \lfloor \log_\eta(2 \eps) \rfloor$ indicates that $\eps_k \le 1/2$ and $\eps_{k+1} > 1/2$.
Then we know $D_{k+1} \le N$ by \cref{clm:count-properties}, \cref{en:count-close-triangle}.
Now that $D_0 \ge M \cdot N$ and $D_{k+1} \le N$, there exists $i \in \set{0, 1, \dots, k}$ such that $D_i \ge M' \cdot \max\{N, D_{i+1}\}$ (for $M' = M^{1/(k+1)}$).
Then setting $\eps' = \eps_i$ and $\vec{t}'$ achieving $N_p(\lat^\dagger, 1 - \eps_i, \vec{t}') = D_i$ satisfies the desired property.
Note that this choice gives $\eps' \in [\eps, 1/2]$.
\end{proof}

We use the following lemma about norm embeddings to extend our gadgets based on exponential kissing number lattices to arbitrary $\ell_p$ norms.

\begin{lemma}[\cite{FLM77}, {\cite[Theorem~3.2]{conf/stoc/RegevR06}}]
\label{lem:norm-embedding}
For any $n \in \Z^+$, $p \in [1, \infty)$, and $\eps > 0$,
there exists $m \in \Z^+$ and a linear map $f: \R^n \to \R^m$,
such that for all $\vec{x} \in \R^n$,
\begin{equation*}
    (1 - \eps) \norm{\vec{x}}_2 \le \norm{f(\vec{x})}_p \le (1 + \eps) \norm{\vec{x}}_2
    \ \text.
\end{equation*}
In particular, it suffices to take $m$ to be
\begin{equation*}
    m = \begin{cases}
        \eps^{-2} n
        & \text{if} \
        1 \le p < 2
        \ \text, \\
        \eps^{-p} (n / p)^{p / 2}
        & \text{if} \
        p \ge 2
        \ \text.
    \end{cases}
\end{equation*}
\end{lemma}

We remark that for fixed $p$, the increased dimensionality of the norm embeddings given by \cref{lem:norm-embedding} is $m = \poly(n, 1 / \eps)$, which is polynomial in $n$ for any $\eps \geq 1 / \poly(n)$.

The following lemma allows us to adapt gadgets in $\ell_2$ norm to the requirements of \cref{thm:generic-bdd}; in particular, the lemma
gives analogous gadgets $(\lat^{\dagger}, \vec{t}^{\dagger})$ in any $\ell_p$ norm by applying the norm embeddings in \cref{lem:norm-embedding} and scaling properly, with the caveat that $\sm_1(\lat^{\dagger})$ and $\dist_p(\vec{t}^{\dagger}, \lat^{\dagger})$ may be distorted by a $1 \pm \eps$ factor.\footnote{Miller and Stephens-Davidowitz~\cite{MSKissing18} referred to the number of non-zero vectors in a lattice $\lat$ of length at most $(1 + \eps) \cdot \sm_1(\lat)$ as its ``$(1+\eps)$-handshake number.'' For our reductions, it essentially suffices to have a family of lattices $\set{\lat_n}_{n=1}^{\infty}$ where $\lat_n$ has rank $n$ and large $(1+\eps)$-handshake number (in the $\ell_p$ norm, with $\eps$ satisying, say, $\eps \leq 1 + O(1/n)$), which is a weaker property than having large kissing number.
Indeed, all of our reductions are just as strong quantitatively (in terms of the relative distance $\alpha$ and rank increase $C$ achieved) using lattices with large handshake number.
The only issue with using lattices with large handshake number (rather than large kissing number) is qualitative. Their use leads us to rely on ``Gap'' variants of (S)ETH. However, in \cref{subsubsec:remove-gap-assumption} we sketch why even this difference is likely not inherent.
}

\begin{lemma}
\label{lem:adapt-gadget}
For any lattice $\lat^\dagger$ of rank $n$, target $\vec{t}^\dagger \in \lspan(\lat^\dagger)$, and constants $0 < \alpha_A \le \alpha_G \le 1$,
if, for some $M_0, M_1 > 0$,
\begin{equation*}
    N_2(\lat^\dagger, \alpha_G, \vec{t}^\dagger) \ge \max \bigl\{ M_0 \cdot N_2^o(\lat^\dagger, 1, \vec{0}), M_1 \cdot N_2^o(\lat^\dagger, \alpha_A, \vec{t}^\dagger) \bigr\}
    \ \text,
\end{equation*}
then for any $p \in [1, \infty)$, there exist a lattice $\lat'$ of rank $n$ and dimension $m = \poly(n)$ and a target $\vec{t}' \in \lspan(\lat')$,
such that
\begin{equation*}
    N_p(\lat', \alpha_G + o(1), \vec{t}') \ge \max \bigl\{ M_0 \cdot N_p^o(\lat', 1, \vec{0}), M_1 \cdot N_p^o(\lat', \alpha_A, \vec{t}') \bigr\}
    \ \text.
\end{equation*}
\end{lemma}

\begin{proof}
Without loss of generality suppose that $\lat^\dagger$ has dimension $n$ so that $\lspan(\lat^\dagger) = \R^n$, and $\vec{0}$ is the closest lattice vector to $\vec{t}^\dagger$ so that $\norm{\vec{t}^\dagger}_2 \le \alpha_G$ (note that $N_2(\lat^\dagger, \alpha_G, \vec{t}^\dagger) \ge M_0 \cdot N_2^o(\lat^\dagger, 1, \vec{0}) \ge M_0 > 0$).
We apply the norm embeddings in \cref{lem:norm-embedding} to the gadget $(\lat^\dagger, \vec{t}^\dagger)$, with error bound $\eps = 1 / n = o(1)$.
Let $(\lat, \vec{t})$ be the embedded gadget.
By \cref{lem:norm-embedding}, we know $\lat$ has rank $n$ and dimension $m = \poly(n)$, and
\begin{align*}
    N_2^o(\lat^\dagger, 1, \vec{0})
    &\ge N_p^o(\lat, 1 - \eps, \vec{0})
    \ \text, \\
    N_2(\lat^\dagger, \alpha_G, \vec{t}^\dagger)
    &\le N_p(\lat, (1 + \eps) \cdot 2 \alpha_G - (1 - \eps) \cdot \alpha_G, \vec{t}) \\
    &= N_p(\lat, (1 + 3 \eps) \cdot \alpha_G, \vec{t})
    \ \text, \\
    N_2^o(\lat^\dagger, \alpha_A, \vec{t}^\dagger)
    &\ge N_p^o(\lat, (1 - \eps) \cdot (\alpha_A + \alpha_G) - (1 + \eps) \cdot \alpha_G, \vec{t}) \\
    &= N_p^o(\lat, (1 - \eps \cdot (1 + 2 \alpha_G / \alpha_A)) \cdot \alpha_A, \vec{t})
    \ \text.
\end{align*}
Then by normalizing with $\lat' = \lat / (1 - \eps')$ and $\vec{t}' = \vec{t} / (1 - \eps')$, where $\eps' = \eps \cdot (1 + 2 \alpha_G / \alpha_A) = o(1)$, we get a desired gadget $(\lat', \vec{t}')$.
\end{proof}

By applying \cref{lem:decrease-close-radius,lem:increase-too-close-radius} to exponential kissing number lattices along with targets $\vec{0}$, and adapting the resulting gadgets to the requirements of \cref{thm:generic-bdd} using \cref{lem:adapt-gadget}, we get the following.

\begin{corollary}
\label{cor:gadget-eth-bdd-kn}
For any $n \in \Z^+$, $p \in [1, \infty)$, and constants $\eps \in (0, 1/2)$, $\delta \in [\eps, 1 - \eps]$ and $\eta \in [2 \eps, 1)$,
there exist $\eps' \in [\eps, 1/2]$, a lattice $\lat^\dagger$ with rank $n$ and dimension $m = \poly(n)$, and a target $\vec{t}^\dagger \in \lspan(\lat^\dagger)$,
such that
\begin{equation*}
    N_p(\lat^\dagger, 1 - \eps' + o(1), \vec{t}^\dagger) \ge (\sin(\theta) \cdot 2^{\kissnum})^{n / (k+1) - o(n)} \cdot \max \bigl\{ N_p^o(\lat^\dagger, 1, \vec{0}), N_p^o(\lat^\dagger, 1 - \eps' / \eta, \vec{t}^\dagger) \bigr\}
    \ \text,
\end{equation*}
where $\theta := \arccos(\frac{\delta^2 + 2 \eps - \eps^2}{2 \delta})$ and $k := \lfloor \log_\eta(2 \eps) \rfloor$.
\end{corollary}

\begin{proof}
By definition of $\kissnum$ (\cref{eq:kiss-num-def}), we know there exists a lattice $\lat'$ of rank $n$ satisfying $\lambda_1(\lat') = 1$ and $N_2(\lat', 1, \vec{0}) = 2^{\kissnum n - o(n)}$.
Then by \cref{lem:decrease-close-radius}, there exists a target $\vec{t} \in \lspan(\lat')$ such that $N_2(\lat', 1 - \eps, \vec{t}) = \sin(\theta)^n / \poly(n) \cdot 2^{\kissnum n - o(n)} = (\sin(\theta) \cdot 2^{\kissnum})^{n - o(n)}$.
Therefore by \cref{lem:increase-too-close-radius}, there exists $\eps' \in [\eps, 1/2]$ and target $\vec{t}' \in \lspan(\lat')$ such that
\begin{equation*}
    N_2(\lat', 1 - \eps', \vec{t}') \ge (\sin(\theta) \cdot 2^{\kissnum})^{n / (k+1) - o(n)} \cdot \max \bigl\{ N_2^o(\lat', 1, \vec{0}), N_2^o(\lat', 1 - \eps' / \eta, \vec{t}') \bigr\}
    \ \text.
\end{equation*}
The result then follows by applying \cref{lem:adapt-gadget} to the gadget $(\lat', \vec{t}')$.
\end{proof}

Using the notation in \cref{thm:generic-bdd},
the family of gadgets given by \cref{cor:gadget-eth-bdd-kn} has $\alpha_G = 1 - \eps' \le 1 - \eps$, $\alpha_A = 1 - \eps' / \eta$, and $\nu_0 = \nu_1 = (\sin(\theta) \cdot 2^{\kissnum})^{1/(k+1)}$.
For $\nu_0, \nu_1 > 1$, we need $\theta > \theta_0 := \arcsin(2^{-\kissnum})$.
By rearranging this is equivalent to $1 - \eps > \sqrt{(\cos(\theta_0) - \delta)^2 + \sin(\theta_0)^2}$, where the right-hand side is minimized by setting $\delta := \cos(\theta_0)$.

Therefore, by taking $1 - \eps$ to be a constant sufficiently close to $\sin(\theta_0) = 2^{-\kissnum}$ and $\eta$ to be a constant sufficiently close to $1$, we can take $\alpha_G$ to be an arbitrary constant greater than $2^{-\kissnum}$ and $\alpha_A$ to be an arbitrary constant less than $\alpha_G$. Plugging such $\alpha_G, \alpha_A$ into \cref{eq:generic-alpha-eth} in \cref{thm:generic-bdd}, we get that it simplifies to $\alpha > 2^{-\kissnum}$.
By applying this family of gadgets to \cref{thm:generic-bdd}, we then get the following.

\begin{corollary}
\label{cor:reduction-eth-bdd-kn}
Assume that $\kissnum > 0$. Then for any $p \in [1, \infty)$, $\gamma > 1$, and $\alpha > 2^{-\kissnum}$,
there exists an efficient non-uniform reduction from $\binGapCVP_{p, \gamma}$ in rank $n$ to $\BDD_{p, \alpha}$ in rank $C n$ for some constant $C > 1$ for all sufficiently large $n$.
\end{corollary}

Combining \cref{cor:reduction-eth-bdd-kn} with the Gap-ETH-hardness result for $\binGapCVP$ in \cref{thm:gapeth-cvp}, we immediately get \cref{thm:gapeth-bdd-kn-informal}, restated below.

\gapethbddkn*

\subsection{Gap-ETH-hardness of BDD from the Integer Lattice}
\label{sec:gapeth-bdd}

In this section we show Gap-ETH-hardness of BDD by instantiating the general reductions in \cref{thm:generic-bdd} with gadgets based on the integer lattice $\Z^n$ as opposed to exponential kissing number lattices. More specifically, we will consider lattice-target gadgets $(\Z^n, t \cdot \vec{1})$.

For $p \in [1, \infty)$, we define
\begin{equation} \label{eq:alpha-eth}
    \alphaeth_p := \inf_{\substack{t \in [0, 1/2], \\ a \ge t}} \frac{a}{\betainv{0}{\betafun{t}{a}}}
    \ \text.
\end{equation}
The numerator in \cref{eq:alpha-eth} corresponds to the resulting gadget's ``close distance,'' and the denominator to the gadget's ``short distance.'' More specifically, the value $b := \betainv{0}{\betafun{t}{a}}$ corresponds to the largest radius $bn^{1/p}$ at which the number of ``short'' integer vectors of norms at most $b n^{1/p}$ is less than the number of ``close'' integer vectors within distance $a n^{1/p}$ of the target $t \cdot \vec{1}$.
We note that $\alphaeth_p \le 1$ since
\begin{equation*}
    \alphaeth_p
    = \inf_{\substack{t \in [0, 1/2], \\ a \ge t}} \frac{a}{\betainv{0}{\betafun{t}{a}}}
    \le \inf_{t \in [0, 1/2]} \lim_{a \to \infty} \frac{a}{\betainv{0}{\betafun{t}{a}}}
    = 1
    \ \text,
\end{equation*}
where the last equality follows from \cref{clm:beta-limits}. 
Moreover, we note that (provably) $\alphaeth_p < 1$ for all $p > 2$ by~\cite[Lemma 5.4]{conf/stoc/AggarwalS18}, and that $\alphaeth_p \to 1/2$ as $p \to \infty$ (see the discussion at the end of this subsection).
The left plot in \cref{fig:alpha-plots} shows explicit values of $\alphaeth_p$ for $p$ around $2$, and how $\alphaeth_p$ compares to the related quantities $\alpha_p^*$ (with norm embeddings) and $\alphakissnum$. 

The following lemma formally analyzes the gadgets based on the integer lattice.

\begin{lemma}
For any $p \in [1, \infty)$ and $\alphaeth_p < \alpha_A < \alpha_G$,
there exist $t \in [0, 1/2]$, $a \ge t$ and $\nu_0, \nu_1 > 1$,
such that for any $n \in \Z^+$ and for $r = a n^{1/p}$,
\begin{equation*}
    N_p(\Z^n, r, t \cdot \vec{1}) \ge \max \bigl\{ \nu_0^{n - o(n)} \cdot N_p^o(\Z^n, r / \alpha_G, \vec{0}), \nu_1^{n - o(n)} \cdot N_p^o(\Z^n, (\alpha_A / \alpha_G) \cdot r, t \cdot \vec{1}) \bigr\}
    \ \text.
\end{equation*}
\end{lemma}

\begin{proof}
By \cref{clm:beta-properties}, \cref{en:beta-approx-Np}, to have the desired inequality, it suffices to have
\begin{equation*}
    \betafun{t}{a}^n \cdot \exp(-C^* \sqrt{n})
    \ge \max \bigl\{ \nu_0^{n - o(n)} \cdot \betafun{0}{a / \alpha_G}^n, \nu_1^{n - o(n)} \cdot \betafun{t}{(\alpha_A / \alpha_G) \cdot a}^n \bigr\}
    \ \text,
\end{equation*}
and then suffices to have
\begin{equation*}
    \betafun{t}{a}
    \ge \max \bigl\{ \nu_0 \cdot \betafun{0}{a / \alpha_G}, \nu_1 \cdot \betafun{t}{(\alpha_A / \alpha_G) \cdot a} \bigr\}
    \ \text.
\end{equation*}
Since $\alpha_A < \alpha_G$, by the monotonicity of $\betafun{t}{\cdot}$, we know that for any $a \ge t$ there exists $\nu_1 > 1$ such that $\betafun{t}{a} \ge \nu_1 \cdot \betafun{t}{(\alpha_A / \alpha_G) \cdot a}$.
Moreover, as $\alpha_G > \alphaeth_p$, by definition of $\alphaeth_p$ in \cref{eq:alpha-eth}, we know there exists $t \in [0, 1/2]$ and $a \ge t$ such that
\begin{equation*}
    \alpha_G > \frac{a}{\betainv{0}{\betafun{t}{a}}}
    \ \text.
\end{equation*}
Hence, by the monotonicity of $\betafun{0}{\cdot}$, there exists $\nu_0 > 1$ so that $\betafun{t}{a} \ge \nu_0 \cdot \betafun{0}{a / \alpha_G}$.
\end{proof}

By applying the normalized gadgets $((\alpha_G / r) \cdot \Z^n, (\alpha_G t / r) \cdot \vec{1})$ to \cref{thm:generic-bdd}, and setting $\alpha_G, \alpha_A$ to be arbitrarily close to $\alphaeth_p$, we get the following.
Note that this family of gadgets is efficiently computable as $(\Z^n, t \cdot \vec{1})$ is efficiently computable and $r$ is set to be $a n^{1/p}$ for some constant $a$, and the associated point counting of the form $N_p(\Z^n, a' n^{1/p}, t' \cdot \vec{1})$ can be efficiently approximated to within a $2^{o(n)}$ factor by computing $\betafun{t'}{a'}^n$.

\begin{corollary}
\label{cor:reduction-eth-bdd-int}
For any $p \in [1, \infty)$, $\alpha > \alphaeth_p$, and $\gamma > 1$, there exists an efficient randomized reduction from $\binGapCVP_{p, \gamma}$ in rank $n$ to $\BDD_{p, \alpha}$ in rank $C n$ for some constant $C > 1$ for all sufficiently large $n$.
\end{corollary}

Combining \cref{cor:reduction-eth-bdd-int} with the Gap-ETH-hardness result for $\binGapCVP$ in \cref{thm:gapeth-cvp}, we immediately get \cref{thm:gapeth-bdd-int-informal}, restated below.

\gapethbddint*

Bennett and Peikert~\cite[Equation 3.5]{conf/coco/BennettP20} define $\alpha^*_p :=
\inf\set{\alpha^* > 0 : \min_{\tau > 0} \exp(\tau / (2 \alpha^*)^p) \cdot \Theta_p(\tau) \leq 2}$ and show an analogous result to \cref{thm:gapeth-bdd-int-informal} for all $\alpha > \alpha_p^*$ (assuming randomized ``non-gap'' ETH). By \cref{clm:beta-properties}, \cref{en:beta-via-min-tau}, we know $\alpha^*_p = 1/(2 \betainv{0}{2})$.
Note that $1 / (2 \betainv{0}{2}) = a / \betainv{0}{\betafun{t}{a}}$ for $a = t = 1/2$.
As a result, we have $\alphaeth_p \le \alpha^*_p$ (see \cref{eq:alpha-eth} for the definition of $\alphaeth_p$), and so \cref{thm:gapeth-bdd-int-informal} gives a result that is at least as strong quantitatively as the corresponding result in \cite{conf/coco/BennettP20} for all $p$ (see the left plot in \cref{fig:alpha-plots}).
Moreover, the upper bounds on $\alpha^*_p$ obtained in \cite[Lemma~3.11]{conf/coco/BennettP20} immediately apply to our $\alphaeth_p$ as well.
In addition, $\betafun{0}{2 a} \geq \betafun{t}{a}$ holds for any $t \in [0, 1/2], a \geq t$ by triangle inequality (see \cref{clm:count-properties}, \cref{en:count-close-triangle}), implying that $\alphaeth_p \geq 1/2$.
Combining these bounds, we get that $\lim_{p \to \infty} \alphaeth_p = 1/2$.

\subsection{Gap-SETH-hardness of BDD from Exponential Kissing Number Lattices}
\label{sec:gapseth-bdd}
In this section we show Gap-SETH-hardness of $\BDD_{p, \alpha}$ using the general reduction established in \cref{sec:reduction-bdd}.
We start by showing how to construct locally dense lattices with exponentially many close vectors, essentially the same close and short distances, and no too-close vectors by sparsifying exponential kissing number lattices. 

\begin{lemma}
\label{lem:kissing-number-gadget}
For all sufficiently large $n \in \Z^+$,
there exists a lattice $\lat^\dagger$ of rank $n$ and a target $\vec{t}^\dagger \in \lspan(\lat^\dagger)$
such that $\lambda_1(\lat^\dagger) \ge 1$, 
$N_2(\lat^\dagger, 1, \vec{t}^\dagger) > \tau^l_n/4$,
and $N^o_2(\lat^\dagger, 1, \vec{t}^\dagger) = 0$.
\end{lemma}

\begin{proof}
By definition, there exists a lattice $\lat$ of rank $n$ with Euclidean kissing number $\tau^l_n$.
Suppose without loss of generality that $\lat$ has dimension $n$ and $\lambda_1(\lat) = 1$.
We will apply \cref{prop:sparsification-lattice-sets} to lattice $\lat$ and target $\vec{t} := \vec{0}$ with index $q = 3$.
In particular, supposing $\mat{B} \in \R^{n \times n}$ is an arbitrary basis of $\lat$, the sparsification samples $\vec{x}, \vec{z} \in \F_q^n$ uniformly at random and sets
\begin{equation*}
    \lat^\dagger := \set{ \vec{v} \in \lat : \iprod{\mat{B}^{-1} \vec{v}, \vec{x}} \equiv 0 \Mod{q}}
    \ \text, \qquad
    \vec{t}^\dagger := \vec{t} -\mat{B} \vec{z} = -\mat{B} \vec{z}
    \ \text.
\end{equation*}
It is easy to see that by construction, $\lambda_1(\lat^\dagger) \ge \lambda_1(\lat) = 1$.
Moreover, note that with $r = \lambda_1(\lat) = 1$, we have $r < q \cdot \lambda_1(\lat) / 2$.
Then by \cref{prop:sparsification-lattice-sets}, \cref{en:sparsification-close}, $N_2(\lat^\dagger, 1, \vec{t}^\dagger) \le \tau^l_n / (q+1)$ with probability at most $q (q+1)^2 / \tau^l_n + 1/q^n$,
and by \cref{prop:sparsification-lattice-sets}, \cref{en:sparsification-too-close}, $N^o_2(\lat^\dagger, 1, \vec{t}^\dagger) > 0$ with probability at most $1 / q + 1/q^n$.
By union bound, $(\lat^\dagger, \vec{t}^\dagger)$ satisfies the desired properties with probability at least
\begin{equation*}
    1 - \frac{q (q+1)^2}{\tau^l_n} - \frac{1}{q} - \frac{2}{q^n}
    \ \text,
\end{equation*}
which is positive for $q = 3$ and all sufficiently large $n$. In particular, $\lat^{\dagger}$, $\tdag$ with the claimed properties exist for all sufficiently large $n$.
\end{proof}

Note that the gadgets constructed in \cref{lem:kissing-number-gadget} have, using the notation in \cref{lem:adapt-gadget}, $\alpha_A = \alpha_G = 1$, $M_0 = \Omega(\tau^l_n) \ge 2^{\kissnum n - o(n)}$ (by definition of $\kissnum$) and arbitrarily large $M_1$.
By applying \cref{lem:adapt-gadget} to the gadgets we immediately get the following.

\begin{corollary}
\label{cor:gadget-bdd-kn}
Assume that $\kissnum > 0$. Then for any $n \in \Z^+$, $p \in [1, \infty)$, and $\nu_1 > 1$,
there exist a lattice $\lat^\dagger$ of rank $n$ and dimension $m = \poly(n)$ and a target $\vec{t}^\dagger \in \lspan(\lat^\dagger)$,
such that
\begin{equation*}
    N_p(\lat^\dagger, 1 + o(1), \vec{t}^\dagger) \ge \max \bigl\{ 2^{\kissnum n - o(n)} \cdot N_p^o(\lat^\dagger, 1, \vec{0}), \nu_1^n \cdot N_p^o(\lat^\dagger, 1, \vec{t}^\dagger) \bigr\}
    \ \text.
\end{equation*}
\end{corollary}

For $p \in [1, \infty)$ and $C > 1$, we define
\begin{equation} \label{eq:alpha-seth}
    \alphaseth_{p, C} := \Bigl( 1 + \frac{1}{(2 \betainv{0}{2^{\kissnum (C - 1)}})^p} \Bigr)^{1/p}
    \ \text,
\end{equation}
which we obtain by plugging $\alpha_G = \alpha_A = 1$, $\nu_0 = 2^{\kissnum}$, and sufficiently large $\nu_1$ into \cref{eq:generic-alpha-seth} in \cref{thm:generic-bdd}. Combining this with \cref{cor:gadget-bdd-kn,thm:generic-bdd}, we get the following.

\begin{corollary}
\label{cor:reduction-seth-bdd}
Assume that $\kissnum > 0$. Then for any constants $p \in [1, \infty)$, $C > 1$, $\gamma > 1$, and $\alpha > \alphaseth_{p, C}$, there exists an efficient non-uniform reduction from $\binGapCVP_{p, \gamma}$ in rank $n$ to $\BDD_{p, \alpha}$ in rank $C n$ for all sufficiently large $n$.
\end{corollary}

Combining \cref{cor:reduction-seth-bdd} with the Gap-SETH-hardness result for $\binGapCVP$ in \cref{thm:gapseth-cvp}, we immediately get \cref{thm:gapseth-bdd-informal}, restated below.

\gapsethbdd*

We note that the extra $1 - \eps$ factor in the runtime lower bound exponent in \cref{thm:gapseth-cvp} gets absorbed by the multiplicative rank increase factor $C$ in \cref{cor:reduction-seth-bdd}.

We remark that for any fixed $C > 1$, $\lim_{p \to \infty} \alphaseth_{p, C} = 1$, and for any fixed $p \in [1, \infty)$, $\lim_{C \to \infty} \alphaseth_{p, C} = 1$.
To show this, we consider the behavior of $\betainv{0}{\cdot}$.
By \cref{clm:beta-0-ub} in \cref{sec:gapseth-svp} we know that $\betafun{0}{a} \le f(a^p)$ for some strictly increasing function $f$ with $\lim_{x \to \infty} f(x) = \infty$ (see \cref{clm:beta-0-ub} for the explicit formula for $f$), and consequently $\betainv{0}{\nu}^p \ge f^{-1}(\nu)$.
Then, recalling the definition of $\alphaseth_{p, C}$ in \cref{eq:alpha-seth}, we have $1 \le \alphaseth_{p, C} \le (1 + 1 / (2^p f^{-1}(2^{\kissnum (C - 1)})))^{1/p}$, which implies the claimed limits.

\subsubsection{An Approach for Using a Weaker Hypothesis}
\label{subsubsec:remove-gap-assumption}

We conclude with an approach for making the result in \cref{thm:gapseth-bdd-informal} depend on a weaker assumption, namely non-uniform SETH rather than non-uniform Gap-SETH. 
The approach considers a reduction akin to the main reduction in \cref{thm:generic-bdd}, which reduces an instance of $\binGapCVP_{p, \gamma}$ of rank $n'$ to an instance of $\BDD_{p, \alpha}$ of rank $Cn' + 1$ for some constant $C > 1$ assuming the existence of a family of locally dense lattices parameterized by $\alpha_G$, $\alpha_A$, $\nu_0$, and $\nu_1$.
\cref{thm:generic-bdd} assumes that all of these parameters are constants, but this requirement is for simplicity and does not seems inherent. Here we will consider what happens if we disregard this requirement and allow $\gamma$, $\alpha_G$, and $\alpha_A$ to depend on $n'$.
The approach uses three observations:
\begin{enumerate}
    \item \cref{thm:seth-cvp} shows hardness of $\binGapCVP_{p, \gamma}$ instances of rank $n'$ with $\gamma \geq 1 + 1/\poly(n')$ assuming (plain) SETH.
    \item Applying norm embeddings (\cref{lem:norm-embedding}) to lattices of rank $n$ with distortion $\eps = \eps(n) \geq 1/\poly(n)$ can be done in $\poly(n)$ time.
    \item Because we can take $\nu_1$ to be arbitrarily large in \cref{cor:gadget-bdd-kn}, when plugging the corresponding gadgets into \cref{thm:generic-bdd} we can make $d_1$ arbitrarily large. In particular, we can make it so that the middle term in \cref{eq:generic-alpha-seth} simplifies to $(\alpha_G^p - \alpha_A^p)/(\gamma^p - 1)$.\footnote{In fact, this is essentially the same observation we already used to derive \cref{cor:reduction-seth-bdd} from \cref{cor:gadget-bdd-kn}.}
\end{enumerate}

Combining these observations, if $\gamma = \gamma(n')$ satisfies $\gamma^p - 1 \geq 1/\poly(n')$---which it does for the SETH-hard $\binGapCVP_{p, \gamma}$ instances from \cref{thm:seth-cvp}---then we can efficiently apply norm embeddings %
with distortion $\eps = o(\gamma^p - 1)$
to obtain modified locally dense lattices of rank $(C - 1)n'$ with the same bounds on $\nu_0$ and $\nu_1$ as in \cref{cor:gadget-bdd-kn}, but with $\alpha_G = \alpha_G(n'), \alpha_A = \alpha_A(n')$ satisfying $\alpha_G^p - \alpha_A^p = o(\gamma^p - 1)$ and hence $(\alpha_G^p - \alpha_A^p)/(\gamma^p - 1) = o(1)$.

Naively plugging these modified locally dense lattices into \cref{eq:generic-alpha-seth} and invoking \cref{thm:seth-cvp} appears to yield a result with the same conclusion as in \cref{thm:gapseth-bdd-informal}, but with the hypothesis weakened to rely only on non-uniform SETH instead of non-uniform Gap-SETH.
However, we emphasize that this is not correct as is since we have violated the assumption that $\gamma$, $\alpha_G$, and $\alpha_A$ are constants in \cref{thm:generic-bdd}.
Nevertheless, it seems very likely that we could remove this assumption (at least in the special case of the locally dense lattices used in this section) and obtain this result.

\section{Gap-SETH-hardness of GapSVP}
\label{sec:gapseth-svp}

The goal of this section is to show Gap-SETH-hardness of $\GapSVP$, stated in \cref{thm:gapseth-svp-informal}, using the integer lattice as a locally dense lattice as in \cref{sec:gapeth-bdd}.
Here the proof idea is again to reduce from $\binGapCVP_{p, \gamma}$, whose Gap-SETH-hardness is known (see \cref{thm:gapseth-cvp}), to $\GapSVP_{p, \gamma'}$.
We start by introducing the intermediate problem $(A, G)$-$\GapCVP_{p, \gamma}$, which was originally defined in~\cite{conf/stoc/AggarwalS18} with a slightly different but equivalent parameterization.

\begin{definition}
For $p \in [1, \infty]$, $\gamma = \gamma(n) \ge 1$, and $A = A(n), G = G(n) \ge 0$, an instance of the decision promise problem $(A, G)$-$\GapCVP_{p,\gamma}$ is a lattice basis $\mat{B} \in \R^{d \times n}$, a target $\vec{t} \in \R^d$, and distances $r, u > 0$.
\begin{itemize}
    \item It is a YES instance if $N_p(\lat(\mat{B}), r, \vec{t}) \ge G$.
    \item It is a NO instance if $A_{p, u}(\lat(\mat{B}), \gamma (r^p + u^p)^{1/p}, \vec{t}) \le A$, where
    \[
        A_{p, u}(\lat, r', \vec{t}) := \sum_{z = 0}^{\lfloor r' / u \rfloor} N_p(\lat, ((r')^p - z^p u^p)^{1/p}, z \vec{t}) - 1
        \ \text.
    \]
\end{itemize}
\end{definition}

The following result from~\cite{conf/stoc/AggarwalS18} gives an essentially rank- and dimension-preserving reduction from $(A, G)$-$\GapCVP_{p,\gamma}$ to $\GapSVP_{p, \gamma}$ for $G \geq 1000 A$. Although we will use it as a black box, we note in passing that it works by combining Kannan's embedding and lattice sparsification, essentially reducing our task to upper bounding the terms $A_{p, u}(\lat, r', \vec{t})$.

\begin{theorem}[{\cite[Theorem 3.2]{conf/stoc/AggarwalS18}}] \label{thm:agcvp-to-svp}
For any $p \in [1, \infty)$, $A = A(n) \geq 1$ and $G = G(n) \geq 1000 A(n)$ computable in $\poly(n)$ time, and $\gamma = \gamma(n) \geq 1$, there is an efficient randomized reduction from $(A, G)$-$\CVP_{p, \gamma}$ instances of rank $n$ and dimension $d$ to $\GapSVP_{p, \gamma}$ instances of rank $n + 1$ and dimension $d + 1$ that runs in time $\poly(d, \log A, \log G)$.
\end{theorem}

So, to show hardness of $\GapSVP_{p, \gamma'}$ for some $\gamma' \geq 1$, it suffices to show hardness of $(A, G)$-$\GapCVP_{p,\gamma'}$ with $G \geq 1000A$.
To do this, we give a reduction from $\binGapCVP_{p, \gamma}$ to $(A, G)$-$\GapCVP_{p,\gamma'}$ similar to the one in~\cite[Corollary 4.2]{conf/stoc/AggarwalS18} but with different choice of parameters and analysis. We note that~\cite[Corollary 4.2]{conf/stoc/AggarwalS18} introduced a similar transformation to the one in \cref{lem:reduction-gadget}, which we will again use.

The main idea behind the new choice of parameters is that scaling the input $\binGapCVP_{p, \gamma}$ instance $(B', \vec{t}')$ by $s = (\eps n)^{1/p}/2$ for some constant $\eps > 0$ makes it so that $s \cdot \dist_p(\vec{t}', \lat(B'))$---which is at most $s$ if $(B', \vec{t}')$ is a YES instance and at least $\gamma s$ if $(B', \vec{t}')$ is a NO instance---is within a multiplicative constant factor of $\dist_p(1/2 \cdot \vec{1}, \Z^n) = n^{1/p}/2$. This makes it so that, applying the transformation in \cref{lem:reduction-gadget} to ($B'$, $\vec{t}'$) to obtain ($B$, $\vec{t}$), we have that $\dist_p(\vec{t}, \lat(B))$ changes by a multiplicative constant factor depending on whether the input $\binGapCVP$ instance is a YES instance or a NO instance.
Similarly, we will set $u := (\eps_u n)^{1/p}/2$ for some constant $\eps_u > 0$.
This choice of parameters allows us to reduce from approximate $\binGapCVP$ to approximate $(A, G)$-$\GapCVP$; the corresponding result~\cite[Corollary 4.2]{conf/stoc/AggarwalS18} takes $s = 1$, $u = 1$ and only works as a reduction to exact $(A, G)$-$\GapCVP$.\footnote{We also note that~\cite[Corollary 4.2]{conf/stoc/AggarwalS18} is only stated as a reduction from $\GapCVP$ instances of a more specific form than $\binGapCVP$. Here we observe that all that's necessary there (respectively, here) is that the input $\GapCVP_p$ instances be $\binGapCVP_p$ instances (respectively, that the input $\GapCVP_{p, \gamma}$ instances be $\binGapCVP_{p, \gamma}$ instances with $\gamma > 1$); $\binGapCVP$ frequently appears in the study of the complexity of lattice problems, and so this shows that the reduction is quite natural.}

\begin{lemma} \label{lem:approx-cvp-to-approx-agcvp}
Let $p \in [1, \infty)$, $\eps > 0$, and $\gamma > 1$ be constants, and let $\eps_u := \gamma^p \eps/(3^p - 1)$. Then for all
\[
1 \leq \gamma' < \big(\frac{1 + \eps_u + \gamma^p \eps}{1 + \eps_u + \eps}\big)^{1/p}
\]
and $n', n \in \Z^+$ satisfying $n' \leq n \leq \poly(n')$,
there is a deterministic Karp reduction from $\binGapCVP_{p, \gamma}$ on lattices of rank $n'$ to $(A, G)$-$\GapCVP_{p, \gamma'}$ on lattices of rank $n$, where
\[
A := c \cdot N_p^o(\Z^n, (1 + \eps_u + \gamma^p \eps)^{1/p} \cdot n^{1/p}/2, \vec{0} ), \qquad G := 2^{n - n'}
\]
for some constant $c = c(p, \eps, \gamma, \gamma')$.
\end{lemma}

\begin{proof}
Let $B'$, $\vec{t}'$ be the input instance of $\binGapCVP_{p, \gamma}$.
We will again use the transformation in \cref{lem:reduction-gadget} with $B^{\dagger} = I_{n - n'}$ and $\vec{t}^{\dagger} = 1/2 \cdot \vec{1}$. Specifically, we set
\[
    \mat{B} := \begin{pmatrix} s \mat{B}' & \mat{0} \\ \mat{I}_{n'} & \mat{0} \\ \mat{0} & I_{n - n'} \end{pmatrix}
    \ \text, \qquad
    \vec{t} := \begin{pmatrix} s \vec{t}' \\ \frac{1}{2} \vec{1} \\ \frac{1}{2} \vec{1} \end{pmatrix}
    \ \text,
\]
with $s := (\eps n)^{1/p}/2$, set $r := (s^p + n/2^p)^{1/p} = ((1 + \eps) \cdot n)^{1/p}/2$, and set $u := (\eps_u n)^{1/p}/2$ (where $\eps_u := \gamma^p \eps/(3^p - 1)$). 
The output $(A, G)$-$\GapCVP_{p, \gamma'}$ instance consists of $\mat{B}$, $\vec{t}$, $r$, and $u$. It is clear that the reduction runs in polynomial time, and remains to prove correctness of the reduction.

Let $\lat = \lat(B)$. Suppose that the input $\binGapCVP_{p, \gamma}$ instance is a YES instance. Then there exists $\vec{x} \in \bit^{n'}$ such that $\norm{B' \vec{x} - \vec{t}'}_p \leq 1$, so applying \cref{lem:reduction-gadget}, \cref{en:gadget-close-lb} we get
\begin{align*}
N_p(\lat, r, \vec{t}) &= N_p(\lat, (s^p + n'/2^p + (n - n')/2^p)^{1/p}, \vec{t}) \\
&\geq N_p(\Z^{n - n'}, (n - n')^{1/p}/2, 1/2 \cdot \vec{1}) \\
&= 2^{n - n'} = G
\ \text.
\end{align*}
It follows that the output $(A, G)$-$\GapCVP_{p, \gamma'}$ instance is a YES instance.

Now, suppose that the input $\binGapCVP_{p, \gamma}$ instance is a NO instance. Then
\begin{align*}
A_{p,u}(\lat, \gamma' \cdot (r^p + u^p)^{1/p}, \vec{t}) &= \sum_{z = 0}^{\lfloor \gamma' (r^p/u^p + 1)^{1/p} \rfloor} N_p(\lat, ((\gamma')^p r^p + ((\gamma')^p- z^p) u^p)^{1/p}, z \cdot \vec{t}) - 1 \\
&= \sum_{z = 0}^{\lfloor \gamma' ((1+\eps + \eps_u)/\eps_u)^{1/p} \rfloor} N_p(\lat, ((1 + \eps + \eps_u) \cdot (\gamma')^p - \eps_u \cdot z^p)^{1/p} \cdot n^{1/p}/2, z \cdot \vec{t}) - 1 \\
&\leq \sum_{z = 0}^{\lfloor \gamma' ((1+\eps + \eps_u)/\eps_u)^{1/p} \rfloor} N_p^o\big(\lat, \big(1 + \gamma^p \eps \big(1 - \frac{z^p - 1}{3^p - 1}\big)\big)^{1/p} \cdot n^{1/p}/2, z \cdot \vec{t} \big)
\ \text,
\end{align*}
where the inequality follows by the strict upper bound on the choice of $\gamma'$ and the definition of $\eps_u$.

We will upper bound summands in the above sum according to three cases. First, consider summands in which $z$ is even. By \cref{lem:reduction-gadget}, \cref{en:gadget-short-ub}, and the fact that $z/2 \cdot \vec{1} \in \Z^n$,
\begin{align}
N_p^o\big(\lat, \big(1 + \gamma^p \eps \big(1 - \frac{z^p - 1}{3^p - 1}\big)\big)^{1/p} \cdot n^{1/p}/2, z \cdot \vec{t} \big)
&\leq N_p^o(\lat, (1 + \eps_u + \gamma^p \eps)^{1/p} \cdot n^{1/p}/2, z \cdot \vec{t}) \nonumber \\
&\leq N_p^o(\Z^n, (1 + \eps_u + \gamma^p \eps)^{1/p} \cdot n^{1/p}/2, z/2 \cdot \vec{1}) \nonumber \\
&= N_p^o(\Z^n, (1 + \eps_u + \gamma^p \eps)^{1/p} \cdot n^{1/p}/2, \vec{0}) \ \text.
\label{eq:Aps-even-z}
\end{align}
Second, consider the case where $z = 1$. Then using the fact that $\dist_p(\vec{t}', \lat(B')) > \gamma$ and applying \cref{lem:reduction-gadget}, \cref{en:gadget-too-close-ub} we have that
\begin{align*}
N_p^o\big(\lat, \big(1 + \gamma^p \eps \big(1 - \frac{z^p - 1}{3^p - 1}\big)\big)^{1/p} \cdot n^{1/p}/2, z \cdot \vec{t} \big)
&= N_p^o(\lat, (1 + \gamma^p \eps)^{1/p} \cdot n^{1/p}/2, \vec{t}) \\
&\leq N_p^o(\Z^n, ((1 + \gamma^p \eps) \cdot n/2^p - \gamma^p s^p)^{1/p}, 1/2 \cdot \vec{1}) \\
&= N_p^o(\Z^n, n^{1/p}/2, 1/2 \cdot \vec{1})
= 0 \ \text,
\end{align*}
where the last equality follows by noting that because $\dist_p(1/2 \cdot \vec{1}, \Z^n) = n^{1/p}/2$ there is no integer point in the open ball of radius $n^{1/p}/2$.
Third, consider the case where $z$ is odd and $z \geq 3$. Then, by \cref{lem:reduction-gadget}, \cref{en:gadget-short-ub} and the fact that $\Z^n - z/2 \cdot \vec{1} = \Z^n - 1/2 \cdot \vec{1}$, we have
\[
\begin{aligned}
N_p^o\big(\lat, \big(1 + \gamma^p \eps \big(1 - \frac{z^p - 1}{3^p - 1}\big)\big)^{1/p} \cdot n^{1/p}/2, z \cdot \vec{t} \big)
&\leq N_p^o(\lat, n^{1/p}/2, z \cdot \vec{t}) \\
&\leq N_p^o(\Z^n, n^{1/p}/2, 1/2 \cdot \vec{1})
= 0
\ \text.
\end{aligned}
\]

The fact that $A_{p,u}(\lat, \gamma' \cdot (r^p + u^p)^{1/p}, \vec{t}) \leq A$ then follows by noting that the sum $A_{p,u}(\lat, \gamma' \cdot (r^p + u^p)^{1/p}, \vec{t})$ has a constant number of terms, that summands in which $z$ is even are upper bounded by the right-hand side in \cref{eq:Aps-even-z}, and that summands in which $z$ is odd are equal to $0$. Hence the output $(A, G)$-$\GapCVP_{p, \gamma'}$ instance is a NO instance.
\end{proof}

It remains to set the parameters so that $G \geq 1000 A$ in \cref{lem:approx-cvp-to-approx-agcvp}. 
Following~\cite{conf/stoc/AggarwalS18},\footnote{Actually,~\cite{conf/stoc/AggarwalS18} defines $C_p$ using the quantity $W_p := \min_{\tau > 0} \exp(\tau/2^p) \Theta_p(\tau)$. However, $W_p$ is equivalent to $\betafun{0}{1/2}$ by \cref{clm:beta-properties}, \cref{en:beta-via-min-tau}, and therefore \cref{eq:Cp-def} is equivalent to the definition of $C_p$ in~\cite{conf/stoc/AggarwalS18}.}
let $p_0 \approx 2.1397$ be the unique solution to the equation $\betafun{0}{1/2} = 2$, and for $p > p_0$ define
\begin{equation}
C_p := \frac{1}{1 - \log_2 \betafun{0}{1/2}}
\ \text.
\label{eq:Cp-def}
\end{equation}
This quantity is essentially the multiplicative rank-increase factor that our reduction from $\binGapCVP_{p, \gamma}$ to $\GapSVP_{p, \gamma'}$ incurs, and hence it determines the quality of the best runtime lower bound on $\GapSVP_{p, \gamma'}$ that we are able show (assuming randomized Gap-SETH).
We note that although we do not have a closed-form expression for $C_p$ itself, we can compute it to high accuracy and also have a rather tight closed-form upper bound on it; see \cref{fig:Cp-plot} and \cref{eq:Cp-closed-form-bound}.

\begin{lemma} \label{lem:reduction-svp}
For every $p > p_0$, $C > C_p$, and $\gamma > 1$, there exists $\gamma' = \gamma'(p, C, \gamma) > 1$ such that there is an efficient randomized reduction from $\binGapCVP_{p, \gamma}$ in rank $n'$ to $\GapSVP_{p, \gamma'}$ in rank $C n' + 1$ for all sufficiently large $n'$.
\end{lemma}

\begin{proof}
By \cref{thm:agcvp-to-svp}, it suffices to reduce $\binGapCVP_{p, \gamma}$ in rank $n'$ with $\gamma > 1$ to $(A, G)$-$\GapCVP_{p, \gamma'}$ in rank $n = C n'$ for some $\gamma' = \gamma'(p, C, \gamma) > 1$ so that $G \geq 1000 A$ for all sufficiently large $n'$.
To do this, by \cref{lem:approx-cvp-to-approx-agcvp}, it suffices to show the existence of $\eps = \eps(p, C, \gamma) > 0$ so that $G \geq 1000 A$ for all sufficiently large $n'$, where
\begin{align*}
    A &= c \cdot N_p^o\big(\Z^n, (1 + \eps_u + \gamma^p \eps)^{1/p} \cdot n^{1/p}/2, \vec{0} \big)
    \leq c \cdot \betafun{0}{(1 + \eps_u + \gamma^p \eps)^{1/p} \cdot 1/2}^n
    \ \text, \\
    G &= 2^{n - n'} = 2^{(1 - 1/C)n}
    \ \text,
\end{align*}
and $\eps_u = \gamma^p \eps/(3^p - 1)$.
(Here the upper bound on $A$ is given by \cref{clm:beta-properties}, \cref{en:beta-approx-Np}.)
Indeed, then we can take $\gamma'$ to be any value in the
range $1 < \gamma' < \big(\frac{1 + \eps_u + \gamma^p \eps}{1 + \eps_u + \eps}\big)^{1/p}$.

By taking $n$-th roots of $A$ and $G$, we see that for $G \geq 1000 A$ to hold for all sufficiently large $n'$, it suffices to have
\begin{equation*}
    \betafun{0}{(1 + \eps_u + \gamma^p \eps)^{1/p} \cdot 1/2} < 2^{1 - 1/C}
    \ \text.
\end{equation*}
Then because $\betafun{0}{a}$ is strictly increasing in $a$, it suffices to have
\begin{equation*}
    (1 + \eps_u + \gamma^p \eps)^{1/p} \cdot 1/2 < \betainv{0}{2^{1 - 1/C}}
    \ \text.
\end{equation*}
By recalling the definition of $C_p$ in \cref{eq:Cp-def}, we know that $\betainv{0}{2^{1 - 1/C}} > 1/2$ holds for $C > C_p$,
and hence there exists $\eps > 0$ satisfying the inequality.
More specifically, one can check that the inequality is satisfied by taking any $\eps$ in the
(non-empty) range
\begin{equation*}
    0 < \eps < \frac{(2 \betainv{0}{2^{1 - 1/C}})^p - 1}{\gamma^p \cdot (1 + 1/(3^p - 1))}
    \ \text.
    \qedhere
\end{equation*}
\end{proof}

Combining \cref{thm:gapseth-cvp} and \cref{lem:reduction-svp} we immediately get \cref{thm:gapseth-svp-informal}, which we restate below.

\gapsethsvp*

We note that because the multiplicative rank increase factor $C$ in \cref{lem:reduction-svp} is \emph{strictly} greater than $C_p$, the $1 - \eps$ factor in the runtime lower bound exponent in \cref{thm:gapseth-cvp} gets absorbed by $C$.

We conclude by showing a closed-form upper bound on $C_p$ based on the following bound on the function $\betafun{0}{a}$.

\begin{claim}
\label{clm:beta-0-ub}
For any $p \in [1, \infty)$ and $a > 0$,
$\betafun{0}{a} \le \exp(\operatorname{arcsinh}(a^p) + \operatorname{arcsinh}(a^{-p}) \cdot a^p)$.
\end{claim}

\begin{proof}
The result is implicit in \cite[Lemma~3.11]{conf/coco/BennettP20}, where it is shown that for any $\sigma > 0$,
\begin{equation*}
    \min_{\tau > 0} \exp(\tau / \sigma) \cdot \Theta_p(\tau, 0) \le \min_\tau \exp(\tau / \sigma) \cdot \Bigl( \frac{2}{1 - \exp(-\tau)} - 1 \Bigr)
    \ \text,
\end{equation*}
and the minimizer for the right-hand side is $\tau = \operatorname{arcsinh}(\sigma)$.
Hence by setting $\sigma = a^{-p}$ we have
\begin{equation*}
\begin{aligned}[b]
    \betafun{0}{a}
    &= \min_{\tau > 0} \exp(\tau a^p) \cdot \Theta_p(\tau, 0) \\
    &\le \exp(\tau a^p) \cdot \Bigl( \frac{2}{1 - \exp(-\tau)} - 1 \Bigr) \Big|_{\tau = \operatorname{arcsinh}(a^{-p})} \\
    &= \exp(\operatorname{arcsinh}(a^p) + \operatorname{arcsinh}(a^{-p}) \cdot a^p)
    \ \text.
\end{aligned}
    \qedhere
\end{equation*}
\end{proof}

By \cref{clm:beta-0-ub} and the definition of $C_p$ in \cref{eq:Cp-def}, we get the following upper bound on $C_p$ for all sufficiently large $p$ (all $p$ greater than approximately $2.2241$ 
so that the denominator is positive):
\begin{equation}
    C_p \le \frac{\ln(2)}{\ln(2) - \operatorname{arcsinh}(1/2^p) - \operatorname{arcsinh}(2^p) / 2^p} 
    \ \text.
    \label{eq:Cp-closed-form-bound}
\end{equation}

\begin{figure}
    \centering
    \includegraphics[scale=0.8]{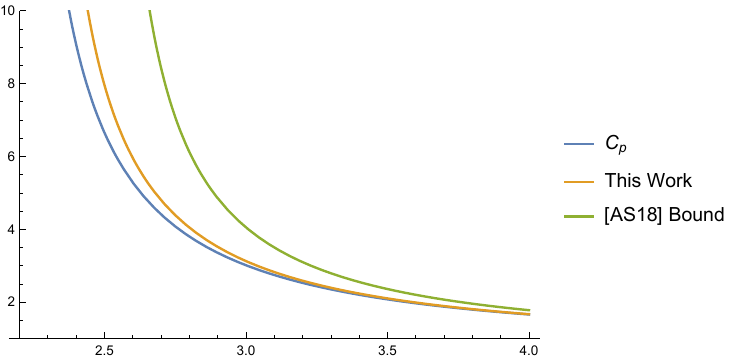}
    \caption{A plot of the multiplicative rank increase factor $C_p$ defined in \cref{eq:Cp-def} and used in \cref{lem:reduction-svp} together with the closed-form upper bounds on $C_p$ given by~\cite[Claim 4.4]{conf/stoc/AggarwalS18} and the bound in this work given in \cref{eq:Cp-closed-form-bound}.}
    \label{fig:Cp-plot}
\end{figure}

We remark that this gives a tighter closed-form bound on $C_p$ than the one in \cite[Claim~4.4]{conf/stoc/AggarwalS18}; see \cref{fig:Cp-plot}.
Furthermore, $\lim_{x \to 0} \operatorname{arcsinh}(x) = 0$ and $\lim_{x \to \infty} \operatorname{arcsinh}(x)/x = 0$, and it is easy to see $C_p \ge 1$, so \cref{eq:Cp-closed-form-bound} shows that $\lim_{p \to \infty} C_p = 1$.

\bibliography{hard-bdd-svp}
\bibliographystyle{alpha}

\appendix
\section{Proofs of Claims in the Preliminaries}
\label{sec:prelims-claims-proofs}

In this appendix, we restate and give proofs for two technical claims from the preliminaries.

\techone*

\begin{proof}%
For \cref{en:beta-approx-Np}, we note that because $a > t$ there exists a unique solution $\tau^*$ to $\mu_p(\tau^*, t) = a^p$ by \cref{clm:rp-equals-mu}. The result then follows from the definition of $\betafun{t}{a}$ and from \cref{thm:point-counts-thetas}.

For \cref{en:beta-via-min-tau}, a calculation shows that $\frac{\partial}{\partial \tau} \ln(\exp(\tau a^p) \cdot \Theta_p(\tau, t)) = a^p - \mu_p(\tau, t)$ and $\frac{\partial^2}{\partial \tau^2} \ln(\exp(\tau a^p) \cdot \Theta_p(\tau, t)) = -\frac{\partial}{\partial \tau} \mu_p(\tau, t) > 0$ (see e.g.\ the proof of \cref{clm:rp-equals-mu} for relevant calculations).
Hence $\tau = \tau^*$, where~$\tau^*$ is the unique solution to $\mu_p(\tau^*, t) = a^p$, is the global minimizer of $\exp(\tau a^p) \cdot \Theta_p(\tau, t)$, and therefore $\betafun{t}{a} = \exp(\tau^* a^p) \cdot \Theta_p(\tau^*, t) = \min_{\tau > 0} \exp(\tau a^p) \cdot \Theta_p(\tau, t)$, as desired.

For \cref{en:beta-continuity}, we start by proving that $\betafun{t}{a}$ is right-continuous at $a = t$, i.e., that $\lim_{a \to t^+} \betafun{t}{a} = 1$ for $t \in [0, 1/2)$ and $\lim_{a \to t^+} \betafun{t}{a} = 2$ for $t = 1/2$.
By the definition of $\betafun{t}{a}$, we have that for every $a > t$ there exists $\tau^* > 0$ such that
\begin{align}
    \betafun{t}{a}
    &= \sum_{z \in \Z} \exp(\tau^* (a^p - \abs{z-t}^p))
    \notag \\
    &= \exp(\tau^* (a^p - t^p)) + \exp(\tau^* (a^p - (1-t)^p)) + \sum_{z \in \Z \setminus \set{0, 1}} \exp(\tau^* (a^p - \abs{z-t}^p))
    \ \text. \label{eq:beta-def-expanded}
\end{align}
By examining the first two terms in \cref{eq:beta-def-expanded}, it is clear that for $a > t$ we have $\betafun{t}{a} > 1$ when $t \in [0, 1/2)$ and $\betafun{t}{a} > 2$ when $t = 1/2$.

We turn to showing that for every $t \in [0, 1/2]$,  $\lim_{a \to t^+} \betafun{t}{a} \le 1$ when $t \in [0, 1/2)$ and $\lim_{a \to t^+} \betafun{t}{a} \le 2$ when $t = 1/2$.
We have that for any $\tau > 0$ and $a \leq (t^p + 1/2)^{1/p}$,
\begin{align}
    \betafun{t}{a}
    &\leq \exp(\tau (a^p - t^p)) + \exp(\tau (a^p - (1-t)^p)) + \sum_{z \in \Z \setminus \set{0, 1}} \exp(\tau (a^p - \abs{z-t}^p)) \notag \\
    &\leq \exp(\tau (a^p - t^p)) + \exp(\tau (a^p - (1-t)^p)) + 2 \sum_{z = 1}^{\infty} \exp(\tau (1/2 - z)) \notag \\
    &= \exp(\tau (a^p - t^p)) + \exp(\tau (a^p - (1-t)^p)) + 2 \cdot \frac{\exp(-\tau / 2)}{1 - \exp(-\tau)} \ \text. \label{eq:beta-ub-small-a}
\end{align}
The first inequality follows from \cref{en:beta-via-min-tau} and \cref{eq:beta-def-expanded}.
The second inequality follows by noting that
\begin{align*}
    a^p - \abs{z - t}^p
    &\leq t^p + 1/2 - \abs{z - t}^p \\
    &\leq t^p + 1/2 - \abs{z}^p - t^p \\
    &\leq 1/2 - \abs{z}^p \\
    &\leq 1/2 - \abs{z}
\end{align*}
for $z \leq -1$, and 
\begin{align*}
    a^p - \abs{z - t}^p
    &\leq t^p + 1/2 - ((z - 1) + (1 - t))^p \\
    &\leq t^p + 1/2 - (z-1)^p - (1 - t)^p \\
    &\leq 1/2 - (z-1)^p \\
    &\leq 1/2 - (z-1)
\end{align*}
for $z \geq 2$, where we used the fact that $\abs{x + y}^p \geq \abs{x}^p + \abs{y}^p$ when $x$ and $y$ are both non-negative or both non-positive.
The equality follows by using the formula for summing geometric series.

The desired bounds for $\lim_{a \to t^+} \betafun{t}{a}$ then follow by taking the limits $\lim_{\tau \to \infty} \lim_{a \to t^+}$ on both sides of \cref{eq:beta-ub-small-a}, where the limits applied to the terms on the right-hand side act as follows:
\begin{align*}
    \lim_{\tau \to \infty} \lim_{a \to t^+} \exp(\tau (a^p - t^p)) &= \lim_{\tau \to \infty} 1 = 1
    \ \text, \\
    \lim_{\tau \to \infty} \lim_{a \to t^+} \exp(\tau (a^p - (1 - t)^p)) &= \begin{cases}
        \lim_{\tau \to \infty} \exp(\tau (t^p - (1 - t)^p)) = 0
        \ \text, &
        t \in [0, 1/2)
        \ \text, \\
        \lim_{\tau \to \infty} 1 = 1
        \ \text, &
        t = 1/2
        \ \text,
    \end{cases} \\
    \lim_{\tau \to \infty} \lim_{a \to t^+} \frac{\exp(-\tau / 2)}{1 - \exp(-\tau)} &= 0
    \ \text.
\end{align*}

We have shown that $\betafun{t}{a}$ is right-continuous at $t = a$ for every $t \in [0, 1/2]$. To finish proving \cref{en:beta-continuity},
it then suffices to show that $\betafun{t}{a}$ is strictly increasing and differentiable for $a > t$.
Let $\tau^* = \tau^*(a)$ be the unique solution to $\mu_p(\tau^*, t) = a^p$.
We calculate the derivative as follows:
\begin{align*}
    \frac{\mathrm{d}}{\mathrm{d} a} \betafun{t}{a}
    &= \frac{\mathrm{d} \tau^*(a)}{\mathrm{d} a} \frac{\partial}{\partial \tau^*} \exp(\tau^* a^p) \cdot \Theta_p(\tau^*, t) + \frac{\partial}{\partial a} \exp(\tau^* a^p) \cdot \Theta_p(\tau^*, t) \\
    &= 0 + \frac{\partial}{\partial a} \exp(\tau^* a^p) \cdot \Theta_p(\tau^*, t) \\
    &= \tau^*(a) \cdot p a^{p-1} \cdot \betafun{t}{a} > 0
    \ \text,
\end{align*}
where the first equality regards $\exp(\tau^* a^p) \cdot \Theta_p(\tau^*, t)$ as a function with two variables $a, \tau^*$ and uses the chain rule to get its total derivative at $(a, \tau^*(a))$,
and the zero term in the second equality follows from \cref{en:beta-via-min-tau}: $\tau^* = \tau^*(a)$ is the minimizer of $\exp(\tau^* a^p) \cdot \Theta_p(\tau^*, t)$ so $\frac{\partial}{\partial \tau^*} \exp(\tau^* a^p) \cdot \Theta_p(\tau^*, t)  = 0$ at $(a, \tau^*(a))$.

\cref{en:beta-inverse-continuity} immediately follows from \cref{en:beta-continuity}.
\end{proof}

\techtwo*

\begin{proof}%
Let $S_k = \sum_{z \in Z} \exp(-\tau \abs{z - t}^p) \cdot \abs{z - t}^{k p}$.
Then $\Theta_p(\tau, t) = S_0$ and $\mu_p(\tau, t) = S_1 / S_0$.
Observe that by definition of integration, we have
\begin{equation*}
    \begin{aligned}
        \lim_{\tau \to 0} S_k \cdot \tau^{k + 1/p}
        &= \lim_{\tau \to 0} \sum_{z \in Z} \exp(-\tau \abs{z - t}^p) \cdot (\tau \abs{z - t}^p)^k \cdot \tau^{1/p} \\
        &= \int_{-\infty}^\infty \exp(-\abs{\zeta}^p) \cdot \abs{\zeta}^{k p} \cdot \mathrm{d} \zeta \\
        &= 2 \int_0^\infty \exp(-x) \cdot x^k \cdot (x^{1/p - 1} / p) \mathrm{d} x \\
        &= 2 \Gamma(k + 1/p) / p
        \ \text.
    \end{aligned}
\end{equation*}
Therefore we immediately get \cref{en:theta-limit} and \cref{en:mu-limit} as
\begin{align*}
    \lim_{\tau \to 0} \Theta_p(\tau, t) \cdot \tau^{1/p}
    &= \lim_{\tau \to 0} S_0 \cdot \tau^{1/p}
    = 2 \Gamma(1/p) / p
    = 2 \Gamma(1 + 1/p)
    \ \text, \\
    \lim_{\tau \to 0} \mu_p(\tau, t) \cdot \tau
    &= \lim_{\tau \to 0} (S_1 \cdot \tau^{1 + 1/p}) / (S_0 \cdot \tau^{1/p})
    = \Gamma(1 + 1/p) /  \Gamma(1/p)
    = 1 / p
    \ \text.
\end{align*}

For \cref{en:beta-limit} , note that $\betafun{t}{a} = \exp(\tau^* a^p) \cdot \Theta_p(\tau^*, t)$ where $\tau^*$ is the unique solution to $\mu_p(\tau^*, t) = a^p$.
By \cref{en:mu-limit}, we know that $\lim_{\tau^* \to 0} \mu_p(\tau^*, t) \cdot \tau^* = 1 / p$ is a constant, and therefore that $\lim_{\tau^* \to 0} \mu_p(\tau^*, t) = \infty$. Thus by the monotonicity of $\mu_p(\tau, t)$ in $\tau$ (see e.g.\ the proof of \cref{clm:rp-equals-mu}), we know that $\lim_{a \to \infty} \tau^* = 0$.
Then $\lim_{a \to \infty} \tau^* a^p = 1 / p$, and, using \cref{en:theta-limit},
\begin{equation*}
    \lim_{a \to \infty} \betafun{t}{a} / a
    = \lim_{a \to \infty} \frac{\exp(\tau^* a^p) \cdot \Theta_p(\tau^*, t) \cdot (\tau^*)^{1/p}}{(\tau^*)^{1/p} \cdot a}
    = \frac{\exp(1/p) \cdot 2 \Gamma(1 + 1/p)}{(1/p)^{1/p}}
    = 2 \Gamma(1 + 1/p) \cdot (e p)^{1/p}
    \ \text.
    \qedhere
\end{equation*}
\end{proof}

We remark that \cref{en:beta-limit} of \cref{clm:beta-limits} can also be derived using the generic fact that for any lattice coset $\lat - \vec{t}$ and centrally symmetric convex body $K$, $\lim_{a \to \infty} \card{(\lat - \vec{t}) \cap aK}/\vol(aK) = 1/\det(\lat)$; here we can set $\lat = \Z^n$ and $K$ to be the unit $\ell_p$ ball in $n$ dimensions.

\end{document}